\documentclass[12pt]{article}
\usepackage[dvips]{graphics,color}
\usepackage{latexsym}
\usepackage{amsmath}
\usepackage{amssymb}
\usepackage{amsfonts,amsthm}

\newtheorem{theorem}{Theorem}[section]
\newtheorem{lemma}[theorem]{Lemma}

\newtheorem{corollary}[theorem]{Corollary}

\newenvironment{remark}{\noindent{\bf Remark:}}{}

\newcommand{\ignore}[1]{}

\makeatletter
\def\argmin{\mathop{\operator@font argmin}}
\makeatother

\begin{document}

\title{Dominator Tree Certification and Independent Spanning Trees\thanks{This work
	is a rewritten and expanded combination of two conference papers,
	``Dominator Tree Verification and Vertex-Disjoint Paths,''
        {\em Proc.~16th ACM-SIAM Symposium on Discrete Algorithms,}
        pp.~433--442, 2005,
        and ``Dominators, Directed Bipolar Orders, and Independent Spanning Trees,''
        {\em Proc.~39th International Colloquium on Automata, Languages and Programming,}
        pp.~375--386, 2012.
    }}

\date{\today}

\author{Loukas Georgiadis$^{1}$ \and Robert E. Tarjan$^{2}$}

\maketitle

\begin{abstract}
How does one verify that the output of a complicated program is correct?  One can formally prove that the program is correct, but this may be beyond the power of existing methods. Alternatively one can check that the output produced for a particular input satisfies the desired input-output relation, by running a checker on the input-output pair. Then one only needs to prove the correctness of the checker. But for some problems even such a checker may be too complicated to formally verify. There is a third alternative: augment the original program to produce not only an output but also a \emph{correctness certificate}, with the property that a very simple program (whose correctness is easy to prove) can use the certificate to verify that the input-output pair satisfies the desired input-output relation.

We consider the following important instance of this general question: How does one verify that the dominator tree of a flow graph is correct?  Existing fast algorithms for finding dominators are complicated, and even verifying the correctness of a dominator tree in the absence of additional information seems complicated. We define a correctness certificate for a dominator tree, show how to use it to easily verify the correctness of the tree, and show how to augment fast dominator-finding algorithms so that they produce a correctness certificate.  We also relate the dominator certificate problem to the problem of finding independent spanning trees in a flow graph, and we develop algorithms to find such trees. All our algorithms run in linear time. Previous algorithms apply just to the special case of only trivial dominators, and they take at least quadratic time.
\end{abstract}

\footnotetext[1]{Department of Computer Science, University of Ioannina, Greece. E-mail: \texttt{loukas@cs.uoi.gr}.}
\footnotetext[2]{Department of Computer Science, Princeton
University, 35 Olden Street, Princeton, NJ, 08540, and Hewlett-Packard Laboratories. E-mail: \texttt{ret@cs.princeton.edu}.
Research at Princeton University partially supported by NSF grants CCF-0830676 and CCF-0832797. Research while visiting Stanford University partially supported by an AFOSR MURI grant. The information contained herein does not necessarily reflect the opinion or policy of the federal government and no official endorsement should be inferred.}

\section{Introduction}
\label{sec:introduction}
A vexing problem when running a computer program is to verify that the output produced by the program is correct. One option is to formally prove the correctness of the program, but if the program implements a sophisticated algorithm, this may be at best challenging and at worst beyond the power of current program verification techniques. A second option is to implement a \emph{checker}~\cite{cheker:BK1995}: an algorithm that, given an input-output pair, verifies that it satisfies the desired input-output relation. For some problems even this may be hard.  There is a third option: augment the original algorithm to produce not only the desired output but also a \emph{correctness certificate}, and implement a \emph{certifier}: an algorithm that, given an input-output pair and a certificate, uses the certificate to verify that the pair satisfies the desired input-output relation.  The augmented algorithm should be as efficient as the original (to within a constant factor), and the certifier should be at least as efficient as the original algorithm, much simpler, and easy to prove correct.

We address the following important instance of this problem: verify the correctness of the dominator tree of a flow graph. A \emph{flow graph} is a directed graph $G$ with a designated \emph{start} vertex $s$ such that every vertex is reachable from $s$. A vertex $v$ \emph{dominates} a vertex $w$ if $v$ is on every path from $s$ to $w$; if $v \not= w$, $v$ is a \emph{strict dominator} of $w$. The dominator relation defines a \emph{dominator tree} $D$ such that $v$ dominates $w$ if and only if $v$ is an ancestor of $w$ in $D$.  (See Figure \ref{fig:dominator-tree}.) The dominator tree is a central tool in program optimization and code generation~\cite{cytron:91:toplas}, and it has applications in other diverse areas including constraint programming~\cite{QVDR:PADL:2006},
circuit testing~\cite{amyeen:01:vlsitest}, theoretical biology~\cite{foodwebs:ab04}, memory profiling~\cite{memory-leaks:mgr2010}, connectivity and path-determination problems~\cite{2vc,2VCSS:Geo,Italiano2012}, and the analysis of diffusion networks~\cite{Rodrigues:icml12}.
Lengauer and Tarjan~\cite{domin:lt} gave two near-linear-time algorithms for computing $D$ that run fast in practice~\cite{dom_exp:gtw} and have been used in many of these applications.  The simpler of these runs in $O(m \log_{(m/n + 1)}{n})$ time on an $n$-vertex, $m$-arc flow graph with $n > 1$. The other runs in $O(m \alpha_{m/n}(n))$ time.
Here $\alpha$ is a functional inverse of Ackermann's function defined as follows: $\alpha_{d}(n) = \min \{ k > 0 \ | \ A(k, \lceil d \rceil) > n \}$, where $A$ is defined recursively by $A(0, j) = j + 1$; $A(k, 0) = A(k - 1, 1)$ if $k > 0$; $A(k, j) = A(k - 1, A(k, j - 1))$ if $k > 0$, $j > 0$.  This is a slight variant of Tarjan's original definition~\cite{dsu:tarjan}. Subsequently, more-complicated but truly linear-time algorithms to compute $D$ were discovered~\cite{domin:ahlt,dominators:bgkrtw,domin:bkrw,dom:gt04}.

\begin{figure}[t]
\begin{center}
\scalebox{0.7}[0.7]{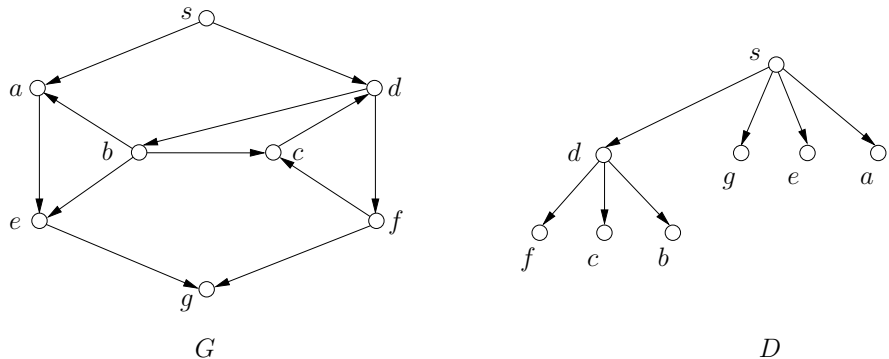}
\end{center}
\caption{\label{fig:dominator-tree} A flow graph and its dominator tree.}
\end{figure}

How does one know that the tree produced by one of these fast but complicated algorithms is in fact the dominator tree?  This natural question was asked of the second author by Steve Weeks in 1999.  Researchers in program verification have asked the same question~\cite{domver:ZZ}. This paper is our answer.

Formal verification of any of the fast algorithms has to our knowledge never been done. Instead, only a much simpler but much slower algorithm has been formally verified \cite{domver:ZZ}. It seems hard even to build a fast checker of the dominator tree. We take the third approach to verification: augment the algorithm that computes the dominator tree to produce in addition a correctness certificate, such that it is simple to verify the correctness of the tree with the help of the certificate. Here what ``simple'' means is not rigorous: since linear time is necessary and sufficient for constructing the tree and hence for verifying it, running time cannot be the measure of simplicity. Our goal is a linear-time certifier that avoids the technical complications of the fast dominator-finding algorithms and is as simple as possible, simple enough that one could easily formally prove the correctness of the certifier. (We leave this task for those skilled in program verification.)

Our certificate is a preorder of the vertices of $D$ with a certain property, which we call \emph{low-high}.  Since a low-high order is a preorder, it can be embedded in the representation of $D$, or it can be represented separately.  Verifying that an order is low-high is entirely straightforward and can be done easily in linear time.  In addition to providing a certifier, we develop linear-time algorithms to compute a low-high order given $G$ and $D$. By adding one of these algorithms to an algorithm to compute $D$, one obtains a \emph{certifying algorithm}~\cite{certifying} to compute $D$.

To obtain our results, we develop new theory about dominators and related concepts, theory that has additional applications.  We obtain our definition of a low-high order by generalizing a notion of Plein~\cite{dir_st:plehn91} and Cheriyan and Reif~\cite{dir_st:cr}, which they used in algorithms for constructing a pair of disjoint spanning trees. Suppose the dominator tree $D$ is flat; that is, each vertex $v \not= s$ has only one strict dominator, $s$.  Then only $s$ and $v$ are common to all paths from $s$ to $v$.  By Menger's Theorem~\cite{menger}, there are two paths from $s$ to $v$ containing no common vertex other than $s$ and $v$.  Whitty~\cite{vdp:whitty} proved that such paths can be realized for all $v$ by a pair of trees: there are spanning trees $B$ and $R$ of $G$, rooted at $s$, such that for any $v$ the paths from $s$ to $v$ in $B$ and $R$ share only $s$ and $v$. We call such trees \emph{disjoint}.  Whitty actually proved something stronger: there are two disjoint spanning trees $B$ and $R$ rooted at $s$ such that, for any distinct vertices $v$ and $w$, either the path in $B$ from $s$ to $v$ and the path in $R$ from $s$ to $w$ share only $s$, or the path in $R$ from $s$ to $v$ and the path in $B$ from $s$ to $w$ share only $s$. We call such trees \emph{strongly disjoint}. Two trees can be disjoint without being strongly disjoint, as the example in Figure \ref{fig:independent-strong} shows.

\begin{figure}
\begin{center}
\scalebox{1}[1]{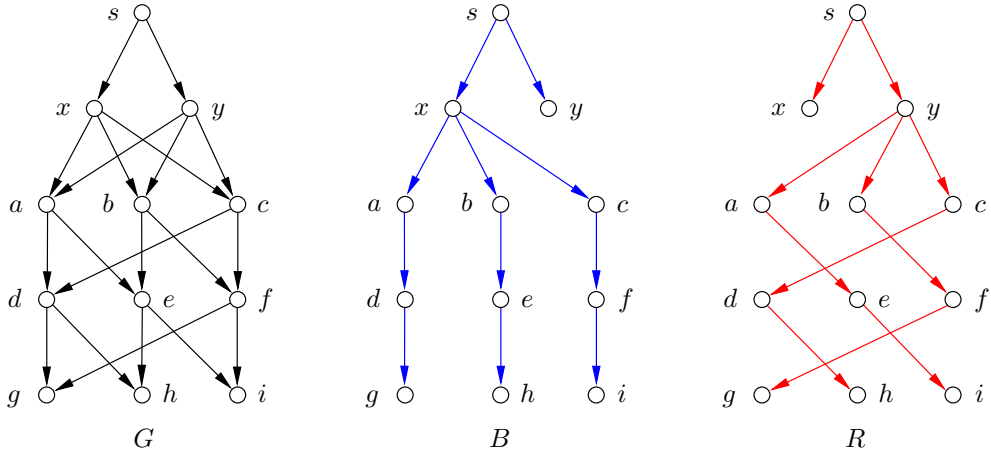}
\end{center}
\caption{\label{fig:independent-strong} A graph with two disjoint but not strongly disjoint spanning trees. Graph $G$ has a flat dominator tree. Spanning trees $B$ and $R$ are disjoint, but the pair $g,h$ violates the condition for strong disjointness.}
\end{figure}

Plehn~\cite{dir_st:plehn91} and independently Cheriyan and Reif~\cite{dir_st:cr} gave simpler proofs of Whitty's result using what Cheriyan and Reif called a directed $st$-numbering as an intermediary.  Given a directed graph with $n$ vertices and two distinct vertices $s$ and $t$, a \emph{directed $st$-numbering} is a numbering of the vertices from $1$ to $n$ such that $s$ is numbered $1$, $t$ is numbered $n$, and every other vertex $v$ has an entering arc from a smaller vertex and an entering arc from a larger vertex.  The proofs of Whitty, of Plehn, and of Cheriyan and Reif give polynomial-time constructions of directed $st$-numberings and of strongly disjoint spanning trees, but their constructions seem to require $\Omega(nm)$ time in the worst case.  Huck~\cite{ind_st:huck94} later gave an $O(nm)$-time algorithm to find two disjoint spanning trees.
None of these proofs or algorithms is especially simple. Note that if $D$ is flat, constructing it is trivial, but $D$ provides no help in constructing two disjoint or strongly disjoint trees. That is, constructing two such trees is a different problem than finding $D$.  Furthermore, verifying that two spanning trees are disjoint or strongly disjoint is not easy.

We extend these concepts to arbitrary flow graphs. Let $G$ be a flow graph with vertex set $V$, arc set $A$, and start vertex $s$. Let $T$ be a tree rooted at $s$ with vertex set $V$ (not necessarily a spanning tree of $G$), with $t(v)$ the parent of $v$ in $T$.
A \emph{preorder} of $T$ is a total order of its vertices obtainable by doing a depth-first traversal of $T$ and ordering the vertices by first visit. A preorder is of $T$ is \emph{low-high} on $G$ if, for all $v \not= s$, $(t(v), v) \in A$ or there are two arcs $(u, v) \in A$, $(w,v) \in A$ such that $u$ is less than $v$, $v$ is less than $w$, and $w$ is not a descendant of $v$. Two spanning trees $B$ and $R$ rooted at $s$ are \emph{independent} if for all $v$, the paths from $s$ to $v$ in $B$ and $R$ share only the dominators of $v$; $B$ and $R$ are \emph{strongly independent} if for every pair of vertices $v$ and $w$, either the path in $B$ from $s$ to $v$ and the path in $R$ from $s$ to $w$ share only the common dominators of $v$ and $w$, or the path in $R$ from $s$ to $v$ and the path in $B$ from $s$ to $w$ share only the common dominators of $v$ and $w$. (See Figure \ref{fig:low-high}.) These three definitions extend the notions of $st$-numberings and disjoint and strongly disjoint spanning trees, respectively, to flow graphs with non-flat dominator trees.  We prove that every flow graph has a pair of strongly independent spanning trees and its dominator tree has a low-high order.

\begin{figure}[t]
\begin{center}
\scalebox{0.7}[0.7]{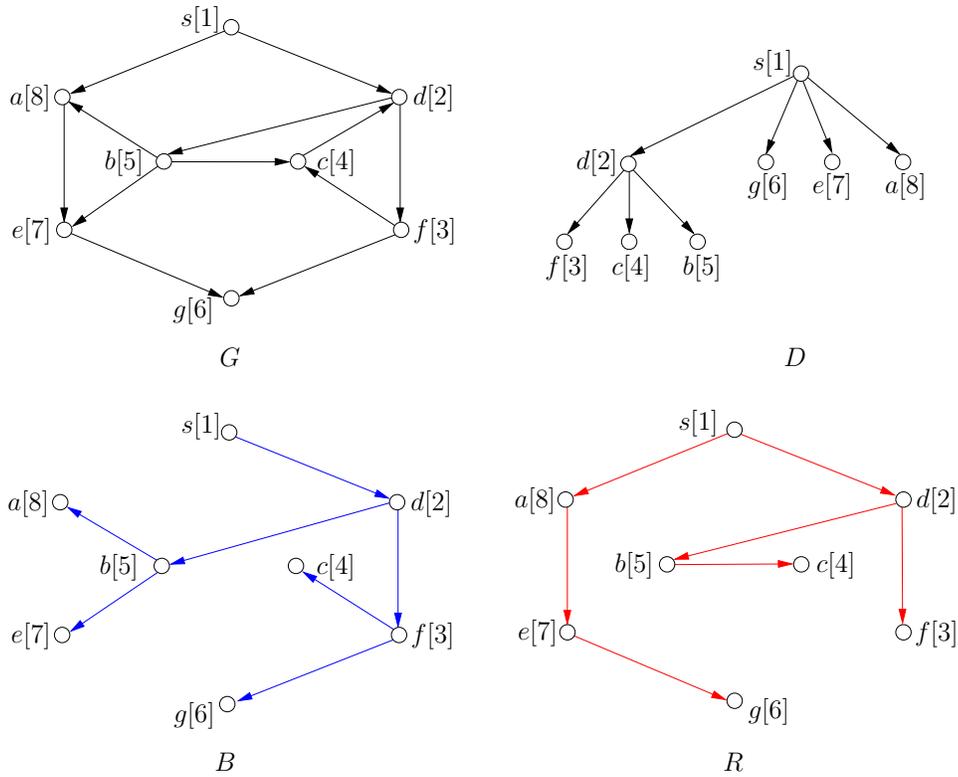}
\end{center}
\caption{\label{fig:low-high} The flow graph of Figure \ref{fig:dominator-tree}, its dominator tree with vertices numbered in low-high order (numbers in brackets), and two strongly independent spanning trees $B$ and $R$.}
\end{figure}

We develop linear-time algorithms related to these concepts, generalizing and significantly improving the previous results. We present a very simple certifier that, given a flow graph and a tree with a vertex order, verifies that the order is low-high and that the tree is the dominator tree of the flow graph.  We present an algorithm that, given a flow graph and its dominator tree with a low-high order, constructs a pair of strongly independent spanning trees.  These trees have the additional property that they are arc-disjoint except for bridges. (A \emph{bridge} is an arc $(u, v)$ such that every path from $s$ to $v$ contains $(u, v)$.) That the trees are strongly independent follows immediately from the fact that the order is low-high.  We present three algorithms to construct a low-high order, given a flow graph and its dominator tree.  The first applies only to reducible flow graphs (defined below), a class that includes acyclic flow graphs.
The second and third apply to arbitrary flow graphs, but both require additional input: the second requires loop-nesting information (defined in Section \ref{sec:low-high-headers}), the third requires information about semi-dominators (defined in Section \ref{sec:semidominators}).  The third algorithm uses the semi-dominator information to construct a pair of independent spanning trees, from which it then constructs a low-high ordering.  Although both steps are simple, the correctness proof of the first step is surprisingly complicated.

Many program control flow graphs are reducible. For these, the first low-high order algorithm can be used to certify the dominator tree.  For arbitrary flow graphs, either the second or third algorithm can be used.  Many applications of dominators need loop information as well; for these, the second algorithm is appropriate, since the loop information must be computed anyway.  The fast algorithms for finding dominators do so by computing semi-dominators, making it easy to extend these algorithms to produce the information needed by the third algorithm.  If the flow graph is not reducible, the application does not need loop information, and dominators are computed by a simple algorithm that does not compute semi-dominators, such as the iterative algorithm of Cooper et al.~\cite{dom:chk01}, the second algorithm can be used if loops are computed, which can be done in near-linear time~\cite{st:t} using disjoint set union~\cite{dsu:tarjan} or truly linear time~\cite{dominators:bgkrtw} using additional techniques.  Thus, depending on the situation, at least one of our three algorithms can be used to certify the dominator tree.

Our paper is a rewritten and expanded combination of two conference papers \cite{domv:gt05,bipolar:2012}. It contains six sections in addition to this introduction. Section \ref{sec:properties} introduces some terminology and develops some properties of dominator trees. It also presents our algorithms to construct a pair of strongly independent spanning trees given a low-high order and to verify the correctness of a dominator tree with a low-high order. Section \ref{sec:derived-graph} introduces the \emph{derived graph}, which in effect allows us to reduce the problem of constructing a low-high order to the case of a flat dominator tree. Section \ref{sec:low-high-reducible} presents our algorithm for finding a low-high order on a reducible flow graph.  As part of the implementation of this algorithm, we develop a simplified linear-time algorithm for a special case of the dynamic list order problem~\cite{list_order:bender_et_al,list_order:ds87}.  This algorithm may itself have additional applications. Sections \ref{sec:low-high-headers} and \ref{sec:semidominators} contain our algorithms to find a low-high order on an arbitrary flow graph.  The algorithm in Section \ref{sec:low-high-headers} uses loop information, and the algorithm in Section \ref{sec:semidominators} uses semi-dominator information. Section \ref{sec:remarks} discusses other applications and open problems.

\section{Dominators, Low-High Orders, and Strongly Independent Spanning Trees}
\label{sec:properties}

\subsection{Terminology}
\label{sec:terminology}

Let $T$ be a rooted tree.  We denote by $t(v)$ the parent of vertex $v$ in $T$; $t(v) = \mathit{null}$ if $v$ is the root of $T$. If $v$ is an ancestor of $w$, $T[v, w]$ is the path in $T$ from $v$ to $w$.  If $v$ is a proper ancestor of $w$, $T(v, w]$ is the path to $w$ from the child of $v$ that is an ancestor of $w$.  Tree $T$ is \emph{flat} if its root is the parent of every other vertex.  A \emph{preorder} of $T$ is a total order of the vertices of $T$ such that, for every vertex $v$, the descendants of $v$ are ordered consecutively, with $v$ first.  Equivalently, for each non-root vertex $v$, $t(v)$ is less than $v$, and if $x$ is a descendant of $v$ but $y$ is not, then $y$ is less than $v$ or greater than $x$. The possible preorders of $T$ are exactly those that can be obtained by totally ordering the children of each vertex, doing a depth-first traversal of $T$, and ordering the vertices in the order they are first visited by the traversal.

Throughout the rest of this paper, $G$ is a flow graph with vertex set $V$, arc set $A$, start vertex $s$, and no arcs entering $s$: arcs entering $s$ can be deleted without affecting any of the concepts we study.  We denote by $n$ the number of vertices and by $m$ the number of arcs.  To simplify bounds we assume $n > 1$.  Since $m \ge  n - 1$, this implies $m = \Omega(n)$.  We denote the dominator tree of $G$ by $D$: if $v \not=  s$, the parent $d(v)$ of $v$ in $D$ is the \emph{immediate dominator} of $v$, the unique strict dominator of $v$ dominated by all strict dominators of $v$. (Some authors use $\mathit{idom}(v)$ to denote the immediate dominator of $v$.)

\subsection{The parent and sibling properties}
\label{sec:parent-sibling-properties}

Let $T$ be a rooted tree whose vertex set is a subset of $V$. Tree $T$ has the \emph{parent property} if for all $(v, w) \in A$, $t(w)$ is an ancestor of $v$ in $T$.  Since $s$ has no entering arcs, but every other vertex has at least one entering arc (all vertices are reachable from $s$), the parent property implies that $T$ is rooted at $s$ and has vertex set exactly $V$. Tree $T$ has the \emph{sibling property} if $v$ does not dominate $w$ for all siblings $v$ and $w$. The parent and sibling properties are necessary and sufficient for a tree to be the dominator tree.

\begin{theorem}
\label{theorem:parent-sibling}
Tree $D$ has the parent and sibling properties.
\end{theorem}
\begin{proof}
Tree $D$ has the sibling property by definition.
Tarjan ~\cite[Lemma 7]{path:tarjan81} proved that $D$ has the parent property. For completeness, we include a proof here.
Suppose $D$ violates the parent property. Then there is an arc $(v, w)$ such that $d(w)$ (the immediate dominator of $w$) is not an ancestor of $v$ in $D$; that is, $d(w)$ does not dominate $v$.  But then there is a path from $s$ to $w$ that avoids $d(w)$, consisting of a path from $s$ to $v$ avoiding $d(w)$ followed by arc $(v, w)$. Thus $d(w)$ does not dominate $w$, a contradiction.
\end{proof}

\begin{lemma}
\label{lemma:parent}
Suppose $T$ has the parent property. Let $u$ and $v$ be vertices such that $v \not= s$ and $u$ is not a descendant of $t(v)$ in $T$. Then any path from $u$ to $v$ contains $t(v)$.
\end{lemma}
\begin{proof}
Consider any path from $u$ to $v$. Let $(x, y)$ be the first arc on this path such that $y$ is a descendant of $t(v)$ in $T$.  By the parent property, $t(y)$ is an ancestor of $x$ in $T$. But since $x$ is not a descendant of $t(v)$ in $T$, it must be the case that $y = t(v)$. That is, the path contains $t(v)$.
\end{proof}

\begin{corollary}
\label{corollary:parent}
Suppose $T$ has the parent property. If $v$ is an ancestor of $w$ in $T$, then $v$ dominates $w$.
\end{corollary}
\begin{proof}
Suppose $v \not= s$.  By Lemma \ref{lemma:parent}, any path from $s$ to $v$ contains $t(v)$, so $t(v)$ dominates $v$. The corollary follows by the transitivity of dominance.
\end{proof}

\begin{corollary}
\label{corollary:parent-2}
Suppose $T$ has the parent property and $v \not= s$.  If $x$ is a vertex on a simple path $P$ from $t(v)$ to $v$, then $x$ is a descendant of $t(v)$ but not a proper descendant of $v$.
\end{corollary}
\begin{proof}
If $x$ were a non-descendant of $t(v)$, the part of $P$ from $x$ to $v$ would contain $t(v)$ by Lemma \ref{lemma:parent}, so $P$ would not be simple. If $x$ were a proper descendant of $v$, the part of $P$ from $t(v)$ to $x$ would contain $v$ by Lemma \ref{lemma:parent}, so $P$ would not be simple.
\end{proof}

\begin{theorem}
\label{theorem:parent-sibling-2}
A tree $T$ has the parent and sibling properties if and only if $T = D$.
\end{theorem}
\begin{proof}
Theorem \ref{theorem:parent-sibling} gives the ``if'' half of the theorem.  Suppose $T$ has the parent property.  By Corollary \ref{corollary:parent}, if $v$ is an ancestor of $w$ in $T$, then $v$ dominates $w$.  Suppose $v$ dominates $w$ but $v$ is not an ancestor of $w$.  Since $v \not= w$, $w$ does not dominate $v$.  By Corollary \ref{corollary:parent}, $w$ is not an ancestor of $v$ in $T$, so $v$ and $w$ are unrelated.  Let $u$ be the nearest common ancestor of $v$ and $w$ in $T$, and let $x$ and $y$ be the children of $u$ that are ancestors of $v$ and $w$.  By Corollary \ref{corollary:parent}, $x$ dominates $v$ and $y$ dominates $w$.  By the transitivity of the dominator relation, $x$ dominates $w$.  Since both $x$ and $y$ dominate $w$, one must dominate the other, which violates the sibling property.  This gives the ``only if'' half of the theorem.
\end{proof}

By Theorem \ref{theorem:parent-sibling-2}, to verify that a tree $T$ is the dominator tree, it suffices to show (1) $T$ is a rooted tree, (2) $T$ has the parent property, and (3) $T$ has the sibling property.  It is straightforward to verify (1) in $O(n)$ time, assuming that $T$ is given by its parent function: we just need to check that the parent function is defined for all vertices other than $s$ and that there are no cycles.  It is also easy to verify (2).  This takes $O(1)$ time per arc, for a total of $O(m)$ time, given an $O(1)$-time test of the ancestor-descendant relation.  There are several simple $O(1)$-time tests of this relation~\cite{dfs:t}. The most convenient one for us is to number the vertices from $1$ to $n$ in any preorder of $T$, and to compute the number of descendants of each vertex $v$, which we denote by $\mathit{size}(v)$.  If vertices are identified by number, $v$ is an ancestor of $w$ if and only if $v \le w < v + \mathit{size}(v)$.  If $T$ is given by its parent function, we can number the vertices and compute their sizes by first building a list of children for each vertex and then doing a depth-first traversal, all of which takes $O(n)$ time.

The hardest step in verification is to show (3).  We shall prove that a tree with the parent property has the sibling property if and only if it has a low-high order. We prove sufficiency in this section and necessity in Section \ref{sec:low-high-reducible} for reducible graphs, and in Section \ref{sec:low-high-headers} for arbitrary graphs.   Unlike the parent property, which is easy to test, it is not so easy to test for the existence of a low-high order, although it is easy to test if a given order is low-high.  Thus we place the burden of constructing such an order on the algorithm that computes the dominator tree, not on the verification algorithm: the order certifies the correctness of the tree.

\subsection{Construction of Two Strongly Independent Spanning Trees}
\label{sec:strongly-independent-spanning-trees}

To prove that a tree with the parent property and a low-high order has the sibling property, and because it is interesting in its own right, we show how to construct a pair of strongly independent spanning trees, given a tree $T$ with the parent property and a low-high order.

\begin{figure}[h]
\begin{center}
\fbox{
\begin{minipage}[h]{16cm}
\begin{center}
\textbf{Algorithm 1: Construction of Two Strongly Independent Spanning Trees $B$ and $R$}
\end{center}
Let $T$ be a tree with the parent property and a low-high order.
For each vertex $v \not= s$, apply exactly one of the following cases to choose arcs $(b(v), v)$ in $B$ and $(r(v), v)$ in $R$ (if Cases 1 and 2 both apply, apply either one):
\begin{description}\setlength{\leftmargin}{10pt}
\item[Case 1:] There are two arcs $(u, v)$ and $(w, v)$ such that $u < v < w$ in low-high order and $w$ is not a descendant of $v$ in $T$. Choose two such arcs, and set $b(v) = u$ and $r(v) = w$.
\item[Case 2:] $(t(v), v)$ is an arc and there is another arc $(u, v)$ such that $u < v$ in low-high order. Choose such an arc, and set $b(v) = u$ and $r(v) = t(v)$.
\item[Case 3:] $(t(v), v)$ is the only arc entering $v$ from a non-descendant of $v$. Set $b(v) = r(v) = t(v)$.
\end{description}
\end{minipage}
}
\end{center}
\end{figure}

Algorithm 1 applied to the graph in Figure \ref{fig:dominator-tree} with the low-high order in Figure \ref{fig:low-high} produces the trees shown in Figure \ref{fig:low-high}.

\begin{lemma}
Let $T$ be a tree with the parent property and a low-high order. For any vertex $v \not= s$, at least one of Cases 1, 2, and 3 applies. Thus Algorithm 1 defines $b(v)$ and $r(v)$ for all $v \not= s$.
\end{lemma}
\begin{proof}
Let $v \not= s$ be a vertex to which Case 1 does not apply. The definition of a low-high order implies that $(t(v), v) \in A$.  If $(t(v), v)$ is the only arc entering $v$ from a non-descendant of $v$, then Case 3 applies.  Otherwise, there is an arc $(u, v)$ such that $u \not= t(v)$ and $u$ is not a descendant of $v$.  Since Case 1 does not apply, $u < v$ in low-high order, so Case 2 applies.
\end{proof}

\begin{lemma}
\label{lemma:R-high-path}
For any vertex $v \not= s$, there is a path in $R$ from $t(v)$ to $v$ containing only $t(v)$ and vertices no less than $v$ in low-high order.
\end{lemma}
\begin{proof}
Suppose the lemma is false.  Let $v$ be the largest vertex in low-high order for which it fails.  Then $r(v) \not= t(v)$, so $r(v) > v$ and $r(v)$ is not a descendant of $v$.  Let $x$ be the sibling of $v$ that is an ancestor of $r(v)$.  Since $r(v) > v$ and the order is a preorder, all descendants of $x$ are greater than $v$.  The choice of $v$ implies that the lemma holds for $x$ and all its descendants.  It follows that there is a path in $R$ from $t(x)$ to $x$ to $r(v)$ that contains only $t(x)$ and vertices no less than $x$. Adding $(r(v), v)$ to this path gives a path from $t(v) = t(x)$ to $v$ satisfying the lemma, a contradiction.
\end{proof}

\begin{theorem}
\label{theorem:B-R-strong}
$B$ and $R$ are strongly independent spanning trees rooted at $s$.
\end{theorem}
\begin{proof}
First we prove that $B$ and $R$ are spanning trees rooted at $s$, then that they are independent, and finally that they are strongly independent. Since each vertex $v \not= s$ has an entering arc $(b(v), v)$ with $b(v) < v$, $B$ is a spanning tree rooted at $s$. By Lemma \ref{lemma:R-high-path}, every vertex is reachable from $s$ in $R$, so $R$ is also a spanning tree rooted at $s$.

Suppose $B$ and $R$ are not independent. Let $v$ be the minimum vertex such that $B[s, v]$ and $R[s, v]$ share a vertex other than a dominator of $v$.  By Corollary \ref{corollary:parent}, $t(v)$ dominates $v$, so $B[s, v]$ and $R[s, v]$ both contain $t(v)$.  By the choice of $v$, $B[s, t(v)]$ and $R[s, t(v)]$ share only dominators of $t(v)$.  By Corollary \ref{corollary:parent-2}, no vertices on $B[s, t(v)]$ or $R[s, t(v)]$ are proper descendants of $t(v)$, but all vertices on $B[t(v), v]$ and $R[t(v), v]$ are descendants of $t(v)$, so $B[t(v), v]$ and $R[t(v), v]$ must share a vertex other than $t(v)$ and $v$.  But $B[t(v), v]$ contains only $v$ and vertices less than $v$, and by Lemma \ref{lemma:R-high-path}, $R[t(v), v]$ contains only $t(v)$ and vertices no less than $v$, so $B[t(v), v]$ and $R[t(v), v]$ share only $t(v)$ and $v$, a contradiction.  Thus $B$ and $R$ are independent.

Suppose that $B$ and $R$ are not strongly independent; let $v$ and $w$ be vertices such that $v < w$ and $B[s, v]$ and $R[s, w]$ share a vertex other than a common dominator of $v$ and $w$.  Let $u$ be the nearest common ancestor of $v$ and $w$ in $T$, which dominates both. Since $B$ and $R$ are independent, $B[s, u]$ and $R[s, u]$ share only dominators of $u$, which are common dominators of $v$ and $w$.
By Corollary \ref{corollary:parent-2}, $B[s,u]$ and $R[s,u]$ contain no proper descendants of $u$, but all vertices on $B[u,v]$ and $R[u,w]$ are descendants of $u$, so $B[s,u]$ and $R[u,w]$ contain only $u$ in common, as do $R[s,u]$ and $B[u,v]$.  If follows that $B[u, v]$ and $R[u, w]$ share a vertex other than $u$. This implies that $v$ and $w$ are unrelated.
Let $x$ and $y$ be the children of $u$ that are ancestors of $v$ and $w$, respectively.  All descendants of $x$ are less than all descendants of $y$.  Path $B[u, v]$ contains only vertices less than $y$, and $R[u, w]$ contains only $u$ and vertices no less than $y$, so $B[u, v]$ and $R[u, w]$ share only $u$, a contradiction.  Thus $B$ and $R$ are strongly independent.
\end{proof}

\begin{remark}
For two distinct vertices $v$ and $w$, the low-high order tells us which pair of paths, $B[s, v]$ and $R[s, w]$, or $R[s, v]$ and $B[s, w]$, share only the common dominators $v$ and $w$: the former if $v < w$, the latter if $v > w$.
\end{remark}
\\

Algorithm 1 runs in $O(m)$ time, given an $O(1)$-time test of the ancestor-descendant relation: for each vertex $v \not= s$, we examine the entering arcs until finding two that allow Case 1 or 2 to be applied, or running out of entering arcs, allowing Case 3 to be applied: $v$ must have at least one entering arc from a non-descendant; if there is exactly one such arc, it must be $(t(v), v)$.  To test the ancestor-descendant relation, we use the $O(1)$-time test described above, with the low-high order serving as the preorder.

Case 3 applies only when $(t(v), v)$ is a bridge (all paths from $s$ to $v$ contain $(t(v), v)$), so $B$ and $R$ are arc-disjoint except for the bridges. Tarjan~\cite{st:t} previously gave a near-linear-time algorithm, subsequently improved to linear time~\cite{dominators:bgkrtw}, to construct a pair of spanning trees that are arc-disjoint except for the bridges.  His construction need not produce independent spanning trees, however. If we are willing to allow $B$ and $R$ to share non-bridges, then we can simplify Algorithm 1 by combining Cases 2 and 3 into one case:
\begin{description}\setlength{\leftmargin}{10pt}
\item[Case 2':] $(t(v), v)$ is an arc. Set $b(v) = r(v) = t(v)$.
\end{description}

\begin{theorem}
\label{theorem:parent-low-high}
If $T$ has the parent property and has a low-high order, then $T$ has the sibling property, and hence $T=D$.
\end{theorem}
\begin{proof}
Apply Algorithm 1 to construct two strongly independent spanning trees $B$ and $R$ rooted at $s$. Let $v$ and $w$ be siblings with common parent $u$ in $T$. Assume without loss of generality that $v<w$. The path $B[s,v]$ avoids $w$, so $w$ does not dominate $v$. The path $B[s,u]$ followed by $R[u,w]$ contains no vertices greater than $u$ and less than $w$ and hence avoids $v$, so $v$ does not dominate $w$. Thus $T$ has the sibling property. By Theorem \ref{theorem:parent-sibling-2} $T=D$.
\end{proof}

A question raised by Algorithm 1 is whether \emph{every} pair of strongly independent spanning trees is the result of applying Algorithm 1 to some low-high order of $D$. The answer is no. Consider the graph $G$ and its spanning trees $B$ and $R$ shown in Figure \ref{fig:low-high-strong}. The dominator tree $D$ of $G$ is flat ($v \not= s$ implies $d(v) = s$),  and $B$ and $R$ are strongly independent.  Deleting arcs in $D$ from $G$ and reversing arcs in $R$ results in the cycle $\Gamma$ shown in Figure \ref{fig:low-high-strong}.  If $B$ and $R$ were constructible by Algorithm 1 from some low-high order of $D$, $\Gamma$ would be acyclic.

\begin{figure}
\begin{center}
\scalebox{1}[1]{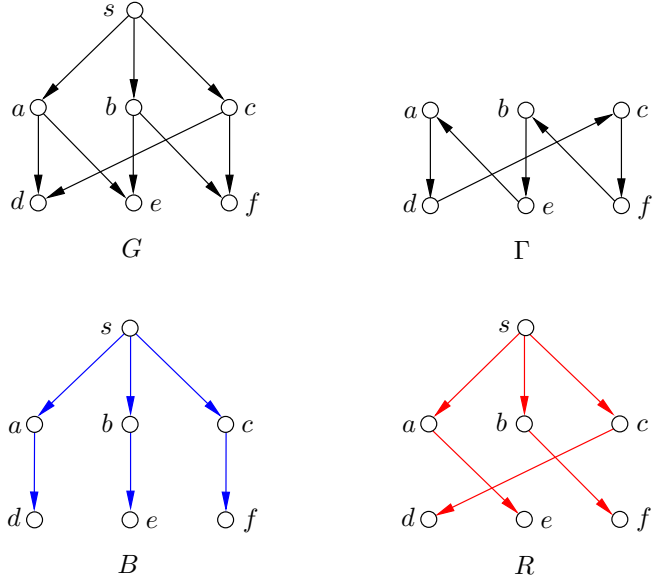}
\end{center}
\caption{\label{fig:low-high-strong} Two strongly independent spanning trees not constructable by Algorithm 1.}
\end{figure}

\subsection{Verification of a Dominator Tree with a Low-High Order}
\label{sec:verification-low-high}

Now we turn to the problem of verifying that a vertex order is low-high.
Let tree $T$ be specified by its parent function, with a vertex order given as a numbering from $1$ to $n$. We can verify that the ordering is low-high as follows. Construct lists of the children of each vertex in increasing order, by beginning with empty lists and adding each vertex $v \not= s$ in increasing order to the back of the list for $t(v)$.  Do a depth-first traversal of $T$ to verify that the corresponding preorder is the same as the given order. During the traversal, compute the size of each vertex. After the traversal, test that the order is low-high by examining the arcs entering each vertex and verifying the existence of the one or two arcs needed to make the order low-high, using the numbers and sizes to test the ancestor-descendant relation in $O(1)$ time.

Thus we obtain the following algorithm to verify a dominator tree with a low-high order:
\begin{figure}[h]
\begin{center}
\fbox{
\begin{minipage}[h]{16cm}
\begin{center}
\textbf{Algorithm 2: Verification of a Dominator Tree $D$ with a Low-High Order}
\end{center}
\begin{description}\setlength{\leftmargin}{10pt}
\item[Step 1:] Verify that the given ``tree'' $D$ is actually a tree.
\item[Step 2:] Construct lists of the children of each vertex in increasing order. Do a depth-first traversal of $D$ to verify that the preorder generated by the search is the same as the given order and to compute the size of each vertex.
\item[Step 3:] Check that the arcs satisfy the parent property and that each vertex has the one or two entering arcs needed to make the order low-high, using the numbers and sizes to test the ancestor-descendant relation in $O(1)$ time.
\end{description}
\end{minipage}
}
\end{center}
\end{figure}

Steps 1 and 2 take $O(n)$ time; Step 3 takes $O(m)$ time.

We conclude this section with a lemma about the structure of strongly connected subgraphs that we shall need in Section \ref{sec:low-high-headers}. A subgraph of $G$ is \emph{strongly connected} if there is a path from any vertex in the subgraph to any other vertex in the subgraph containing only vertices in the subgraph.

\begin{lemma}
\label{lemma:parent-strongly-connected}
Let $T$ be a tree with the parent property, and let $S$ be the set of vertices of a strongly connected subgraph of $G$. Then $S$ consists of a set of siblings in $T$ and possibly some of their descendants in $T$.
\end{lemma}
\begin{proof}
Let $x$ be the nearest common ancestor in $T$ of the vertices in $S$. If $x \in S$ the lemma holds: $S$ contains $x$ and possibly some of its descendants in $T$.  If $x \not \in S$, there are at least two children of $x$ in $T$ that have vertices in $S$ as descendants, and all vertices in $S$ are descendants of such children.  Any path from a non-descendant of such a child of $x$ to a descendant must contain $x$ by the parent property. Thus any such child must be in $S$.
\end{proof}

\section{The Derived Graph}
\label{sec:derived-graph}

Theorem \ref{theorem:parent-low-high} states that a tree with the parent property and a low-high order has the sibling property and hence is $D$.  It remains to prove the converse, namely that $D$ has a low-high order, and to develop fast algorithms to find such an order.  We do this in the next three sections.  In this section we introduce the \emph{derived graph}~\cite{path:tarjan81}, which makes our task easier by in effect reducing the problem of finding a low-high order to the case of a flat dominator tree.

Let $T$ be a tree with the parent property.  (Tree $T$ could be $D$ or a tree claimed to be $D$.)  By the definition of the parent property, if $(v, w)$ is an arc, the parent $t(w)$ of $w$ is an ancestor of $v$ in $T$. The \emph{derived arc} of $(v, w)$ is null if $w$ is an ancestor of $v$, $(v', w)$ otherwise, where $v' = v$ if $v = t(w)$, $v'$ is the sibling of $w$ that is an ancestor of $v$ if $v \not= t(w)$.   The \emph{derived graph} $G'$ is the graph with vertex set $V$ and arc set $A' = \{ (v', w) \ | \  (v', w) \mbox{ is the non-null derived arc of some arc in } A \}$.  See Figure \ref{fig:derived-graph}. This definition differs from the original \cite{path:tarjan81} in that it omits self-loops (arcs of the form $(v,v)$), which are irrelevant for our purposes.

\begin{figure}[t]
\begin{center}
\scalebox{0.8}[0.8]{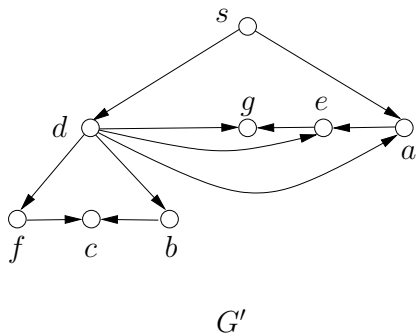}
\end{center}
\caption{\label{fig:derived-graph} The derived graph of the flow graph of Figure \ref{fig:dominator-tree}.}
\end{figure}

\begin{lemma}
\label{lemma:derived-parent}
Tree $T$ has the parent property in $G'$.
\end{lemma}
\begin{proof}
Each derived arc leads from a vertex to one of its children or siblings.
\end{proof}

\begin{lemma}
\label{lemma:derived-simple-path}
If $P$ is a simple path from $s$ to $v$ in $G$, there is a simple path from $s$ to $v$ in $G'$ containing only vertices on $P$.
\end{lemma}
\begin{proof}
By induction on the number of arcs on $P$.  The lemma is immediate if $P$ has no arcs.  Suppose $P$ has at least one arc.  Let $(u, v)$ be the last arc on $P$.  Since $P$ is simple, $u$ cannot be a descendant of $v$ by Corollary \ref{corollary:parent}. Thus $(u, v)$ has a non-null derived arc $(u', v)$.  Since $u'$ is an ancestor of $u$ in $T$, $u'$ is on $P$ by Corollary \ref{corollary:parent}.  By the induction hypothesis, there is a simple path $P'$ from $s$ to $u'$ in $G'$ containing only vertices on the part of $P$ from $s$ to $u'$.  Adding $(u', v)$ to this path gives the desired path $P'$.
\end{proof}

\begin{corollary}
\label{corollary:derived-flowgraph}
Graph $G'$ is a flow graph.
\end{corollary}
\begin{proof}
By Lemma \ref{lemma:derived-simple-path}, in $G'$ all vertices are reachable from $s$.
\end{proof}

\begin{lemma}
\label{lemma:derived-arc}
If $(v', w)$ is the non-null derived arc of an arc $(v, w)$, there is a path in $G$ from $v'$ to $w$ containing only $w$ and descendants of $v'$ in $T$.
\end{lemma}
\begin{proof}
By Corollary \ref{corollary:parent}, $v'$ dominates $v$.  Since $G$ is a flow graph, there is a path $P$ from $s$ to $v$.  Path $P$ contains $v'$.  By Corollary \ref{corollary:parent-2}, the part of $P$ from $v'$ to $v$ contains only descendants of $v'$ in $T$.  Adding $(v, w)$ to the end of this path gives a path satisfying the lemma.
\end{proof}

\begin{corollary}
\label{corollary:derived-acyclic}
If $G$ is acyclic, so is $G'$.
\end{corollary}
\begin{proof}
If $G'$ contained a cycle, so would $G$ by Lemma \ref{lemma:derived-arc}.
\end{proof}

\begin{lemma}
\label{lemma:derived-sibling}
Tree $T$ has the sibling property in $G$ if and only if it has the sibling property in $G'$.  That is, $T$ is the dominator tree of $G$ if and only if it is the dominator tree of $G'$.
\end{lemma}
\begin{proof}
Let $v$ and $w$ be siblings in $T$.  Suppose there is a simple path $P$ from $s$ to $w$ in $G$ that avoids $v$.  Then the path $P'$ from $s$ to $w$ in $G'$ given by Lemma \ref{lemma:derived-simple-path} also avoids $v$.  Suppose there is a simple path $P'$ from $s$ to $w$ in $G'$ that avoids $v$.  Let $u$ be the common parent of $v$ and $w$.  By Corollary \ref{corollary:parent-2}, there is a path in $G$ from $s$ to $u$ that contains no proper descendants of $u$ and hence avoids $v$.  Thus all we need to show is that there is a path from $u$ to $w$ in $G$ that avoids $v$.  The part of $P'$ from $u$ to $w$ consists of an arc from $u$ to a child of $u$, followed by a sequence of arcs from one child of $u$ to another, no such child being $v$.  The first such arc is also an arc of $G$.  Consider one of the derived arcs $(x', y)$ on $P'$ from one child of $u$ to another.  By Lemma \ref{lemma:derived-arc}  there is a path in $G$ from $x'$ to $y$ that contains only $y$ and descendants of $x'$ in $T$, and hence avoids $v$.  Thus we can replace each arc of the part of $P'$ from $u$ to $w$ by a $v$-avoiding path in $G$.  Therefore there is a path in $G$ from $s$ to $w$ that avoids $v$.
\end{proof}

\begin{lemma}
\label{lemma:derived-preorder}
A preorder of $T$ is low-high on $G$ if and only if it is low-high on $G'$.
\end{lemma}
\begin{proof}
Let $v \not= s$.  Suppose the order is low-high on $G$.  By Theorem \ref{theorem:parent-low-high}, $T = D$.  Let $v \not= s$.  If $(t(v), v) \in A$, then $(t(v), v) \in A'$.  If $(t(v), v) \not \in A$, there are arcs $(u, v)$ and $(w, v)$ such that $u < v < w$ in the given preorder and $w$ is not a descendant of $v$ in $T$.  Vertex $u$ is not a descendant of $v$ since it is numbered less than $v$.  Thus $(u, v)$ and $(w, v)$ have non-null derived arcs $(u', v)$ and $(w', v)$, $u'$ and $w'$ are siblings of $v$ in $T$, and $u' < v < w'$ in the given preorder.  Thus the order is low-high on $G'$.

Suppose the order is low-high on $G'$.  Let $v \not= s$.  If $(t(v), v) \in A'$, then $(t(v), v) \in A$.  If $(t(v), v) \not \in A'$, there are arcs $(u, v)$ and $(w, v)$ in $G$ with non-null derived arcs $(u', v)$ and $(w', v)$ such that $u' < v < w'$ in the given preorder and $u'$ and $v'$ are siblings of $v$ in $T$.  Since the order is a preorder, $u < v < w$ in the preorder.  Furthermore $w$ is not a descendant of $v$.  Thus the order is low-high on $G$.
\end{proof}

We can find the derived arcs in $O(m)$ time with Algorithm 3.\ignore{the following algorithm:}
\vspace{2ex}
\begin{figure}[h]
\begin{center}
\fbox{
\begin{minipage}[h]{16cm}
\begin{center}
\textbf{Algorithm 3: Construction of Derived Arcs}
\end{center}
\begin{description}\setlength{\leftmargin}{10pt}
\item[Step 1:] Let $T$ be a tree with the parent property.  Number the vertices of $T$ from $1$ to $n$ in preorder and compute the size of each vertex. Identify vertices by number.
\item[Step 2:] Delete each arc $(v, w)$ such that $v$ is a descendant of $w$ in $T$.
\item[Step 3:] For each arc $(v, w)$ such that $v = t(w)$, let $(v, w)$ be its derived arc.
\item[Step 4:] For each arc $(v, w)$ such that $v$ and $w$ are unrelated, construct a triple $(t(w), v, w)$.  For each vertex $u > 1$, construct a triple $(t(u), u, 0)$.
\item[Step 5:] Sort the triples in increasing lexicographic order by doing a three-pass radix sort.
\item[Step 6:] Process the triples in increasing order.  To process a triple of the form $(t(u), u, 0)$, set $x = u$, where $x$ is a global variable.  To process a triple of the form $(t(w), v, w)$, let $(x, w)$ be the derived arc of $(v, w)$.
\end{description}
\end{minipage}
}
\end{center}
\end{figure}

\begin{theorem}
\label{theorem:derived-graph-construction}
Algorithm 3 is correct.
\end{theorem}
\begin{proof}
Steps 2 and 3 correctly handle the arcs $(v, w)$ such that $v$ and $w$ are related in $T$.  Consider an arc $(v, w)$ such that $v$ and $w$ are unrelated in $T$.  Let $(v', w)$ be the derived arc of $(v, w)$.  Then $t(v') = t(w)$ and $v' \le v$, so triple $(t(v'), v', 0)$ precedes $(t(w), v, w)$ once the triples are sorted lexicographically.  Suppose there is a triple $(t(u), u, 0)$ following $(t(v'), v', 0)$ but preceding $(t(w), v, w)$.  Then $t(u) = t(v')$ and $v' < u \le v$.  But vertices are numbered in preorder, so $u$ must be a descendant of $v'$, a contradiction.  Hence there is no such triple $(t(u), u, 0)$, which implies $x = v'$ when $(t(w), v, w)$ is processed, so Step 6 correctly computes the derived arc of $(v, w)$.
\end{proof}

\section{Reducible Flow Graphs}
\label{sec:low-high-reducible}

A \emph{reducible} flow graph~\cite{rdflow:hu74,reducibility:jcss:tarjan} is one in which every strongly connected subgraph $S$ has a single \emph{entry} vertex $v$ such every path from $s$ to a vertex in $S$ contains $v$.  There are many equivalent characterizations of reducible flow graphs~\cite{reducibility:jcss:tarjan}, and there are algorithms to test reducibility in near-linear~\cite{reducibility:jcss:tarjan} and truly linear~\cite{dominators:bgkrtw} time.  One notion of a ``structured'' program is that its flow graph is reducible.  Reducibility simplifies many computations, although not the computation of dominators, as far as we can tell.  A flow graph is reducible if and only if it becomes acyclic when every arc $(v, w)$ such that $w$ dominates $v$ is deleted~\cite{reducibility:jcss:tarjan}.  Deletion of such arcs does not change the dominator tree, since no such arc can be on a simple path from $s$.  Deleting such arcs thus reduces the problem of finding a low-high order on a reducible flow graph to the same problem on an acyclic graph.  Such a graph has a topological order (a total order such that if $(x, y)$ is an arc, $x$ is ordered before $y$) \cite{fundamental-algorithms:Knuth}.

The following lemma provides a way to find the arcs needed to satisfy the low-high property on a reducible graph:

\begin{lemma}
\label{lemma:derived-sibling-arcs}
Let $T$ be a tree with the parent property, and let $G'$ be the corresponding derived graph. Suppose $T$ has the sibling property in $G$ (or in $G'$).  If $v \not= s$, then $(t(v), v) \in A$ (and in $A'$) or $v$ has in-degree at least two in $G'$.
\end{lemma}
\begin{proof}
Since $T$ has the sibling property, $T = D$.  Suppose the lemma is false for some $v \not= s$.  Since $v$ is reachable from $s$ by a simple path, there is an arc $(u, v)$ such that $v$ does not dominate $u$.  Let $(u', v)$ be the derived arc of $(u, v)$.  Since the lemma is false for $v$, $u \not= d(v)$, and there is no other arc $(w, v)$ with a derived arc $(w', v)$ such that $w' \not= u'$.  But then $u'$ dominates $v$, contradicting the sibling property.
\end{proof}

Lemma \ref{lemma:derived-sibling-arcs} holds for arbitrary graphs.  For reducible graphs, the condition in Lemma \ref{lemma:derived-sibling-arcs} is not only necessary but sufficient for a tree $T$ with the parent property to have the sibling property:

\begin{lemma}
\label{lemma:derived-reducible-sufficient}
Suppose $G$ is reducible.  Let $T$ be a tree with the parent property, and let $G'$ be the corresponding derived graph.  Then $T$ has the sibling property if, for all vertices $v \not= s$, $(t(v), v) \in A$ or $v$ has in-degree at least two in $G'$.
\end{lemma}
\begin{proof}
Delete all arcs $(x, y)$ such that $y$ dominates $x$.  This does not change the dominators, and it makes the graph acyclic.  Suppose the lemma is false.  Let $u$, $v$ be a pair of siblings such that $u$ dominates $v$, with $v$ minimum in some topological order (an order such that if $(x, y)$ is an arc, $x$ is ordered before $y$).  There is a path from $s$ to $t(v)$ that avoids $u$ by Corollary \ref{corollary:parent-2}.  Thus if $(t(v), v) \in A$, $u$ does not dominate $v$, a contradiction.  If $(t(v), v) \not \in A$, there is an arc $(x, v)$ in $G$ with a derived arc $(x', v)$ such that $x'$ is a sibling of $u$ and $v$ but $x' \not= u$.  By the choice of $v$, $u$ does not dominate $x'$, so it cannot dominate $v$, a contradiction.
\end{proof}

\subsection{Low-high orders on reducible flow graphs}
\label{sec:low-high-orders-reducible}

To find a low-high order on a reducible graph, we delete arcs $(x, y)$ such that $y$ dominates $x$, making the graph acyclic.  We then construct for each vertex an ordered list of its children in $D$ by processing the vertices other than $s$ in topological order and inserting each vertex into its set of siblings in a position determined by the arc or arcs whose existence is guaranteed by Lemma \ref{lemma:derived-sibling-arcs}. After building the ordered lists of children, we obtain a low-high order by doing a depth-first traversal of $D$.

\ignore{This approach gives us the following algorithm:}
\begin{figure}[h]
\begin{center}
\fbox{
\begin{minipage}[h]{16cm}
\begin{center}
\textbf{Algorithm 4: Construction of a Low-High Order on a Reducible Flow Graph}
\end{center}
\begin{description}\setlength{\leftmargin}{10pt}
\item[Step 1:] Delete every arc $(v, w)$ such that $w$ is an ancestor of $v$ in $D$, and find the derived arc $(v', w)$ of each remaining arc $(v, w)$ with respect to $D$.
\item[Step 2:] For each vertex $v$, initialize its list of children $C(v)$ to be empty.  Apply the following step to each vertex $v \not= s$ in a topological order on $G$ (or on $G'$): If $(d(v), v) \in A$, insert $v$ anywhere in $C(d(v))$.  Otherwise, find derived arcs $(u', v)$ and $(w', v)$ such that $u' \not= w'$.  Insert $v$ just before $u'$ in $C(d(v))$ if $w'$ is before $u'$ in $C(d(v))$, just after $u'$ otherwise.
\item[Step 3:] Do a depth-first traversal of $D$, visiting the children of each vertex $v$ in their order in $C(v)$.  Number the vertices from $1$ to $n$ as they are visited.  The resulting order is low-high on $G$.
\end{description}
\end{minipage}
}
\end{center}
\end{figure}

This approach gives us Algorithm 4.
Figure \ref{fig:reducible-low-high} shows how this algorithm works.

\begin{figure}[t]
\begin{center}
\scalebox{0.7}[0.7]{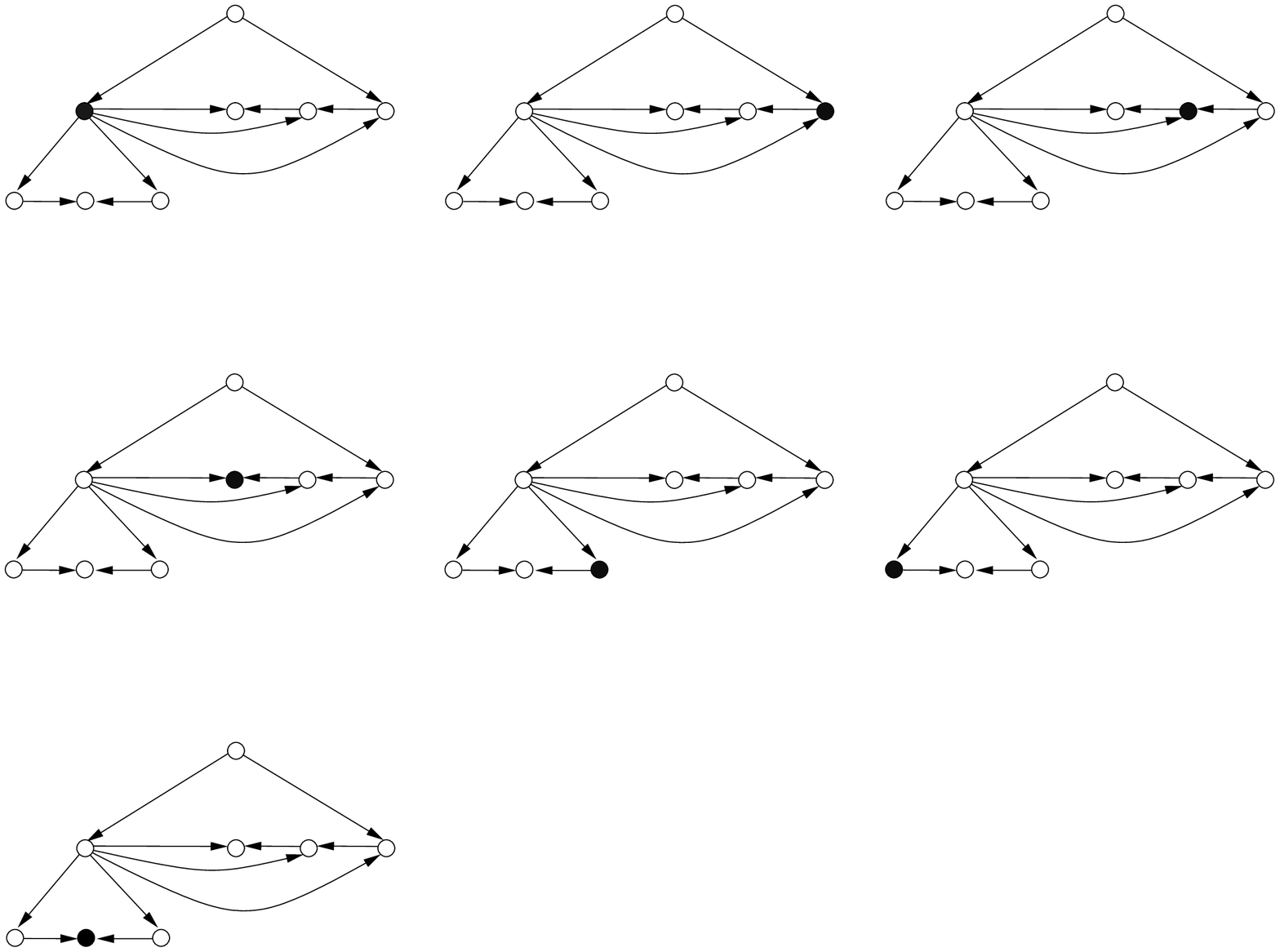}
\end{center}
\caption{\label{fig:reducible-low-high} Computation of a low-high order starting with an acyclic derived graph. Vertices are processed in the topological order $(s,d,a,e,g,b,f,c)$; the vertex processed at each application of Step 2 is shown filled.}
\end{figure}

\begin{theorem}
\label{theorem:low-high-reducible}
Algorithm 4 is correct.
\end{theorem}
\begin{proof}
After Step 1, $G$ is acyclic and every arc has a derived arc. If $(v, w)$ is an arc with derived arc $(v', w)$, $v'$ is an ancestor of $v$, so there is a path from $v'$ to $v$.  It follows that when a vertex $v$ is processed in Step 2, $x'$ is in $C(d(v))$ for every derived arc $(x', v)$ such that $x' \not= d(v)$.  By the condition in Lemma \ref{lemma:derived-reducible-sufficient}, either $(d(v), v)$ is an arc or there are derived arcs $(u', v)$ and $(w', v)$ such that $u' \not= w'$.  Thus the insertion of $v$ into $C(v)$ will succeed.  Step 4 produces a preorder of $T$.  Let $v$ be any vertex other than $s$.  If $(d(v), v) \not \in A$, there are arcs $(u, v)$ and $(w, v)$ in $A$ with derived arcs $(u', v)$ and $(w', v)$, respectively, such that $v$ is ordered between $u'$ and $w'$, and $u'$ and $w'$ are siblings in $T$.  Since the order is a preorder, $v$ is ordered between $u$ and $w$.  Thus the order is low-high.
\end{proof}

When using Algorithm 4 in combination with the dominator certification algorithm (Algorithm 2), we can omit Step 3 in Algorithm 4, since Step 2 in Algorithm 2 does the same tree traversal and numbering. That is, it suffices to represent the low-high order by the ordered lists of children produced by Step 2.\\

\begin{remark}
If we want to verify the dominator tree of a reducible graph but are not interested in constructing a low-high order, we can do the verification using Lemma \ref{lemma:derived-reducible-sufficient}. Given $D$, we verify that $D$ is a tree with the parent property as in Section \ref{sec:properties}, compute the derived graph as in Section \ref{sec:derived-graph} (verifying the correctness of the derived graph using ancestor-descendant tests), and verify the condition in Lemma \ref{lemma:derived-reducible-sufficient}.  The total time for verification by this method is $O(m)$.
\end{remark}\\

It is easy to implement most of Algorithm 4 to run in $O(m)$ time.  In Step 1 we find the arcs to be deleted using an $O(1)$-time ancestor-descendant test as discussed previously, and we find the derived arcs using Algorithm 3.  In Step 2, finding a topological order takes $O(m)$ time using either successive deletion of vertices of in-degree zero~\cite{fundamental-algorithms:Knuth,all-topological-orders:KS} or depth-first search~\cite{dfs:t}.  Finding the needed arcs in Step 2 takes $O(1)$ time per arc, for a total of $O(m)$ time.  Step 3 takes $O(m)$ time.  The only hard part is constructing the lists of children in Step 2.

\subsection{Off-line dynamic list maintenance}
\label{sec:list-order}

For constructing the lists of children, we need a data structure that maintains a list subject to insertions and \emph{order queries}: given $x$ and $y$ in the list, which occurs first? This is the \emph{dynamic list maintenance} problem.  There are solutions to this problem that support insertion and order tests in $O(1)$ time, either amortized or worst-case~\cite{list_order:bender_et_al,list_order:ds87}.  Unfortunately, these solutions are rather complicated, especially those with an $O(1)$ worst-case time bound.  Fortunately, we only need a solution to a special case of dynamic list maintenance, in which there are no deletions and, more importantly, the sequence of operations is given off-line in an appropriate sense.  For this version of the problem there is a simple solution, which we now describe.

Given an initial list with no items, we want to perform off-line an intermixed sequence of the following three kinds of instructions, and then number the items in the final list consecutively from 1.

\begin{list}{}{
\setlength{\leftmargin}{1.7cm} \setlength{\labelsep}{.2cm} \setlength{\itemsep}{0cm}\setlength{\labelwidth}{1.5cm}
}
\item[\emph{after}$(x,y)$:] Return \textbf{true} if $x$ is after $y$ in the current list, \textbf{false} otherwise.  Items $x$ and $y$ must be in the current list.
\item[\emph{insert}$(x)$:] Insert $x$ anywhere in the list.
\item[\emph{insert}$(x,y,\mathit{test})$:] If \emph{test} is true, insert $x$ just after $y$; otherwise, insert $x$ just before $y$.  Item $y$ but not $x$ must be in the current list; \emph{test} is a Boolean combination of \textbf{true} and a fixed number of $\mathit{after}$ queries.
\end{list}

Algorithm 5 gives a simple and efficient solution to this problem.

\begin{figure}[h]
\begin{center}
\fbox{
\begin{minipage}[h]{16cm}
\begin{center}
\textbf{Algorithm 5: Off-Line Execution of a Sequence of Insert and After Instructions}
\end{center}
\begin{description}\setlength{\leftmargin}{10pt}
\item[Step 1 (off-line):] Construct a rooted tree with one non-root vertex for each item ever in the list, and a root.  The parent of item $x$ is the root if there is an instruction $\mathit{insert}(x)$, or $y$ if there is an instruction $\mathit{insert}(x, y, \mathit{test})$.  For each node $x$, let $\mathit{size}(x)$ be the number of nodes in its subtree.
\item[Step 2 (on-line):] Execute the instructions in sequence while maintaining an interval $[i, j]$ for the root and for each item currently in the list.  Here $i$ and $j$ are integers such that $i \le j$.  These intervals will be disjoint.  Initially the root has interval $[0, \mathit{size}(\mathit{root})]$.  Given a query $\mathit{after}(x, y)$, answer \textbf{true} if the interval for $x$ follows that of $y$, \textbf{false} otherwise.  Given an instruction $\mathit{insert}(x)$ such that the root has interval $[i, j]$, replace the interval for the root by $[i, j - \mathit{size}(x)]$ and give $x$ the interval $[j - \mathit{size}(x) + 1, j]$.  (This corresponds to inserting $x$ first in the list.)  Given an instruction $\mathit{insert}(x, y, \mathit{test})$ such that $y$ has interval $[i, j]$, if test is true replace the interval for $y$ by $[i, j - \mathit{size}(x)]$ and give $x$ the interval $[j - \mathit{size}(x) + 1, j]$; if test is false give $x$ the interval $[i, i + \mathit{size}(x) - 1]$ and replace the interval for $y$ by $[i + \mathit{size}(x), j]$.  (The former corresponds to inserting $x$ after $y$, the latter to inserting $x$ before $y$.)
\item[Step 3:]  After Step 2, the root has interval $[0, 0]$, and each item $x$ has an interval $[i, i]$ for some $i > 0$, $i$ distinct for each $x$. If item $x$ has interval $[i, i]$, assign number $i$ to $x$. This numbers the items consecutively from $1$ in list order.
\end{description}
\end{minipage}
}
\end{center}
\end{figure}

\begin{theorem}
\label{theorem:low-high-list}
Algorithm 5 is correct.
\end{theorem}
\begin{proof}
We can view Step 2 as starting with the tree constructed in Step 1, cutting the arc $(\mathit{root},x)$ when $\mathit{insert}(x)$ occurs, and cutting the arc $(y, x)$ when $\mathit{insert}(x, y, \mathit{test})$ occurs.  The items currently in the list are exactly the roots of the trees into which the initial tree has been cut, excluding the initial root.  Each arc cut connects a root with a child.  The interval of a root contains exactly as many integers as the number of vertices in its current tree, plus one if it is the initial root.  It follows that each interval $[i, j]$ has $i \le j$, which guarantees that the implementation is correct.
\end{proof}

We use Algorithm 5 to implement Step 2 of Algorithm 4 as follows.  For each $v$, we construct a sequence of the operations that add vertices to $C(v)$, as follows:  For each addition of a vertex $x$ to $C(v)$ in an arbitrary position, we construct an operation $\mathit{insert}(x)$. For each addition of a vertex $x$ before or after another vertex $y$ depending on whether $\mathit{test}$ is \textbf{true}, we construct an operation $\mathit{insert}(x, y, \mathit{test})$.
Each such test is a single $\mathit{after}$ query; there are no other queries. We do the list operations using Algorithm 5.
At the end of Algorithm 5, the items in $C(v)$ will be numbered consecutively in list order.
Use of Algorithm 5 eliminates the need to use a complicated on-line dynamic list order algorithm. Algorithm 5 may have other applications as well.

\section{Low-High Orders from Loop Nesting Information}
\label{sec:low-high-headers}

In this section we extend the low-high ordering algorithm of Section \ref{sec:low-high-reducible} to arbitrary flow graphs. To do so we must overcome the circularity caused by cycles. Ramalingan~\cite{loops:Ramalingam} has shown how to reduce the problem of computing dominators on an arbitrary flow graph to computing dominators on an acyclic graph.  His reduction runs in $O(m \alpha_{m/n}(n))$ time and can be improved to run in $O(m)$ time: it is an extension of the computation of a loop nesting forest that we discuss below.  Ramalingan's reduction changes the graph and hence affects the low-high order. For this reason we do not use his reduction but base our construction directly on the loop nesting forest.
The rough idea is to repeatedly contract strongly connected subgraphs to single vertices until no cycles exist. Then we apply the algorithm of Section \ref{sec:low-high-reducible} to the resulting acyclic graph, but when a vertex corresponding to a contracted subgraph is to be processed, we expand the subgraph and process its vertices, recursively expanding each such vertex that itself corresponds to a contracted subgraph.

To obtain a sequence of subgraphs to construct, we use the notion of a \emph{loop nesting forest}.  There are several ways to define such a forest \cite{loops:Havlak,loops:Maurer,loops:Ramalingam,loops:SGL,loops:Steensgaard,st:t}.
The one we use was first presented by Tarjan~\cite{st:t} and later rediscovered by Havlak~\cite{loops:Havlak}. It is a particular case of more general definitions \cite{loops:Maurer,loops:Ramalingam}. Let $F$ be the spanning tree generated by a depth-first search of $G$ starting from $s$, with $f(v)$ the parent of vertex $v$ in $F$.  Every cycle in $G$ contains an arc from a descendant to an ancestor in $F$~\cite{dfs:t}. Such an arc is called a \emph{back arc}.  Number the vertices from $1$ to $n$ in reverse postorder with respect to the search, and identify vertices by number. (Postorder, also known as \emph{finishing order}, is the order in which the search finishes its vertex visits.  See \cite{dfs:t}.)  An arc is a back arc if and only if it leads from a larger to a smaller vertex \cite{dfs:t}.  Deleting the back arcs makes the graph acyclic, and makes the vertex order given by the numbering topological \cite{dfs:t}.

The \emph{head} $h(v)$ of a vertex $v$ is the maximum proper ancestor $u$ of $v$ in $F$ such that there is a path from $v$ to $u$ containing only descendants of $u$ in $F$; if there is no such $u$, $h(v) = \mathit{null}$. The heads define a \emph{loop nesting forest} $H$: $h(v)$ is the parent of $v$ in $H$.  Graph $G$ is acyclic if and only if all heads are null.  For each vertex $v$ that is not a leaf in $H$, the descendants of $v$ in $H$ induce a non-trivial strongly connected subgraph of $G$, called a loop.  Because $H$ is a tree, any two loops are either disjoint or one contains the other.  Figure \ref{fig:headers} illustrates these concepts.  In general $H$ depends on $F$.  Graph $G$ is reducible if and only if every $F$ defines the same $H$.

\begin{figure}
\begin{center}
\scalebox{.8}[.8]{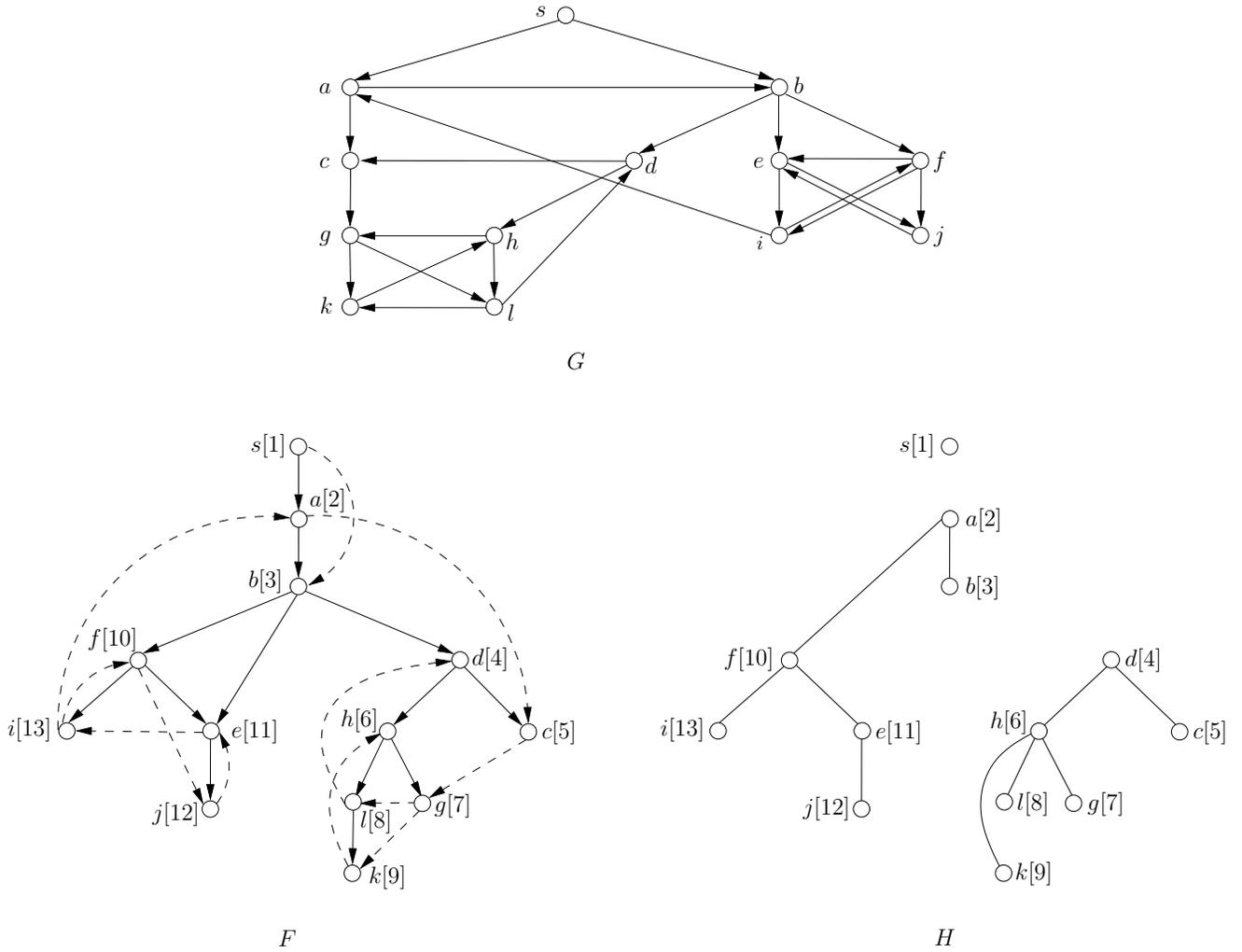}
\end{center}
\caption{A graph, a depth-first spanning tree (tree arcs are solid, non-tree arcs are dashed) with vertices numbered in reverse postorder (in brackets), and the corresponding loop nesting forest. There are five loops, with heads $e$, $f$, $a$, $h$, $d$.\label{fig:headers}}
\end{figure}

Forest $H$ defines a contraction sequence as follows. For each vertex $v$ in decreasing order, if $v$ is not a leaf in $H$, contract the subgraph induced by $v$ and all its children into a single vertex $v$. This subgraph is the \emph{interval} of $v$. The interval of $v$ is strongly connected (just before it is contracted), and all its vertices are descendants of $v$ in $F$. If all arcs leaving $v$ are deleted from its interval, the interval becomes acyclic, and $v$ is the unique vertex with no outgoing arc. Thus, if one starts from any vertex in the interval other than $v$ and follows any path in the interval, one eventually reaches $v$ without repeating a vertex. See Figure \ref{fig:headers-2}. Note that intervals are defined in the contracted graph, whereas loops are defined in the original graph.

\begin{figure}
\begin{center}
\scalebox{.8}[.8]{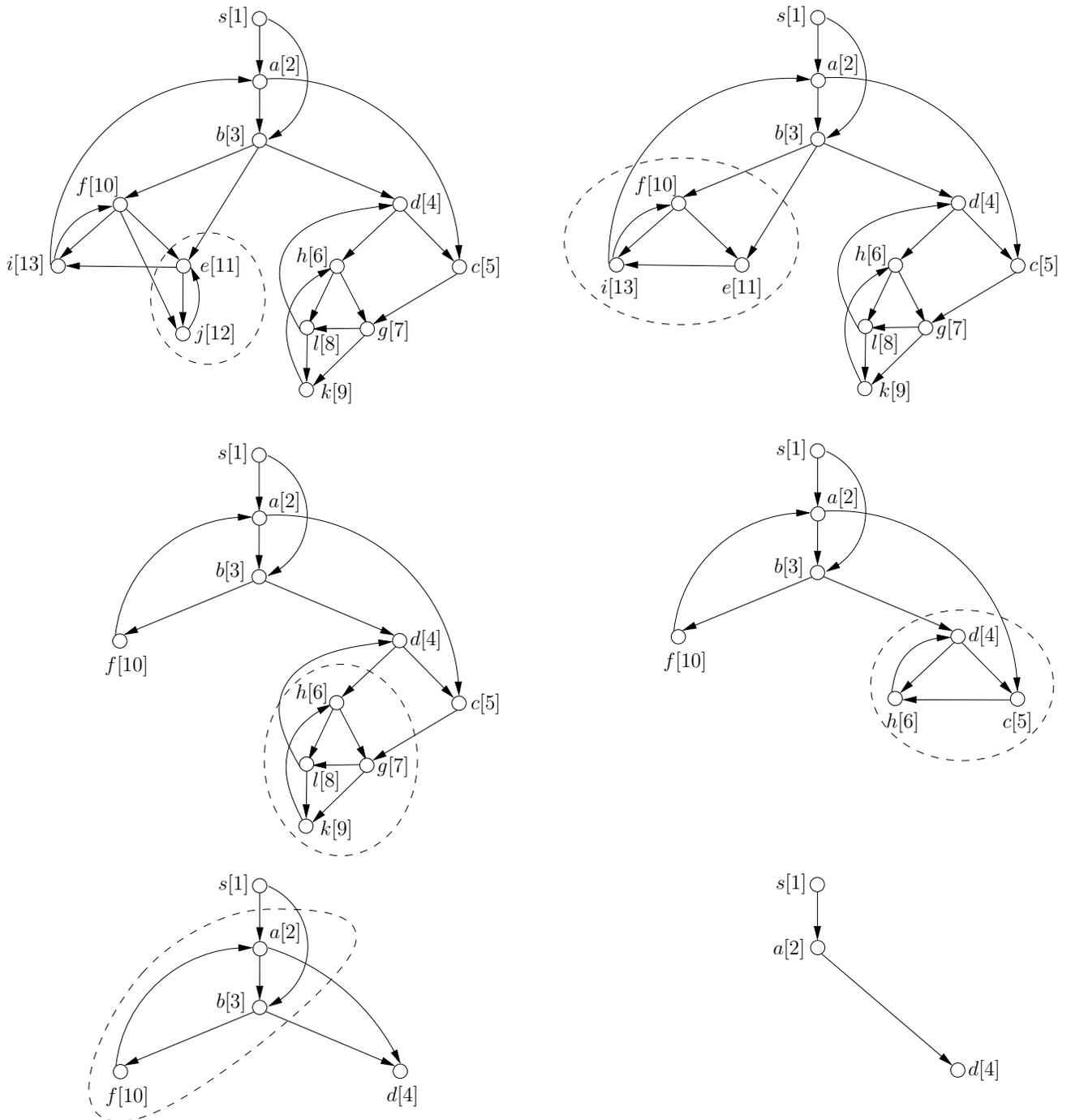}
\end{center}
\caption{Contraction sequence and intervals (strongly connected subgraphs) corresponding to the loop nesting forest in Figure \ref{fig:headers}.\label{fig:headers-2}}
\end{figure}

Many applications of dominators need a loop nesting forest as well \cite{dominators:bgkrtw,fdom:G,Italiano2012,loops:Ramalingam,st:t}. Such a forest can be computed in $O(m \alpha_{m/n}(n))$ time (\cite{st:t}, later rediscovered \cite{loops:Ramalingam1999}) or by a more-complicated algorithm in $O(m)$ time \cite{dominators:bgkrtw}.  These algorithms require less machinery than the fast algorithms for finding dominators (see \cite{dominators:bgkrtw}), but they are not entirely straightforward.  It is easy to extend these algorithms to find, for each vertex $v$ that is not a root of $H$, an outgoing arc from $v$ in the interval of $h(v)$, along with the corresponding original arc; indeed, the forest-construction algorithms proceed by traversing such outgoing arcs backward, doing so by examining the corresponding original arcs. We shall assume that a loop nesting forest, associated outgoing arcs, and the corresponding original arcs are available; if not, they can be computed by one of the cited algorithms.

\begin{lemma}
\label{lemma:headers}
Let $u$ be a vertex, and let $(v, w)$ be an arc such that $w$ is a descendant of $u$ in $H$.  Then $v$ is a descendant of $u$ in $H$ if and only if $v \ge u$.
\end{lemma}
\begin{proof}
Since all descendants of $u$ in $H$ are also descendants of $u$ in $F$, $w > u$.  If $v > w$, then $(v, w)$ is a back arc and $v$ is a descendant of $w$ in $H$ and hence a descendant of $u$ in $H$.  Suppose $v < w$.  If $v \ge u$, $v$ is a descendant of $u$ in $F$ and hence a descendant of $u$ in $H$.  Suppose $v < u$.  Then $v$ is not a descendant of $u$ in $F$, and hence cannot be a descendant of $u$ in $H$.
\end{proof}

\begin{lemma}
\label{lemma:headers-parent}
Let $T$ be a tree with the parent property. Let $u \not= s$, and let $(v, w)$ be an arc such that $w$ but not $v$ is a descendant of $u$ in $H$. Then $w$ is $u$ or a sibling of $u$ in $T$.
\end{lemma}
\begin{proof}
By the parent property, if $x$ is an ancestor of $y$ in $T$, then $x$ dominates $y$. This implies $x$ is an ancestor of $y$ in $F$, so $x < y$.  The set $S$ of descendants of $u$ in $H$ induces a strongly connected subgraph of $G$. By Lemma \ref{lemma:parent-strongly-connected}, $S$ consists of a set of siblings and possibly some of their descendants in $T$. Since $v$ is not such a descendant, the parent property implies that $w$ is one of these siblings. Since $u$ is minimum in $S$, $u$ is also one of these siblings.
\end{proof}

To guarantee the existence of arcs to satisfy the low-high property, we need an analogue of Lemma \ref{lemma:derived-sibling-arcs}. For each vertex $u \not= s$, we denote by $(f'(u), u)$ the derived arc of $(f(u), u)$, which is non-null because the path in $F$ from $s$ to $f(u)$ avoids $u$.

\begin{lemma}
\label{lemma:headers-sibling}
Let $T$ be a tree with the parent property. Then $T$ has the sibling property if and only if, for each $u \not= s$, $f(u) = t(u)$ or there is an arc $(y, w)$ with derived arc $(y',w)$ such that $w$ is a descendant of $u$ in $H$, $y < u$, and $y' \not= f'(u)$.
\end{lemma}
\begin{proof}
If $x$ dominates $z$, $x$ must be an ancestor of $z$ in $F$, so $x \le z$.  Thus if arc $(x, v)$ has a derived arc $(x', v)$, $x' \le x$.  Suppose $T$ has the sibling property.  Let $u \not= s$ be such that $f(u) \not= t(u)$.  Then $f'(u)$ is a sibling of $u$ in $T$.  By the sibling property, $f'(u)$ does not dominate $u$, so there is a path from $s$ to $u$ avoiding $f'(u)$.  Let $(y, w)$ be the first arc on this path with $w$ a descendant of $u$ in $H$.  By Lemma \ref{lemma:headers}, $y < u \le w$, so $w$ does not dominate $y$, which means that $(y, w)$ has a derived arc $(y', w)$.  Since $y'$ is on the path from $s$ to $w$, $y' \not= f'(u)$. Thus $T$ satisfies the condition in the lemma.

Conversely, suppose $T$ satisfies the condition in the lemma.  Suppose there are siblings $v$, $u$ such that $v$ dominates $u$.  Choose such a pair with $u$ minimum.  By Corollary \ref{corollary:parent-2} there is a path from $s$ to $t(u)$ that avoids $v$.  If $f(u) = t(u)$ then $v$ does not dominate $u$.  If $f(u) \not= t(u)$, by the condition in the lemma there is an arc $(y, w)$ with derived arc $(y', w)$ such that $w$ is a descendant of $u$ in $H$, $y < u$, and $y' \not= f'(u)$.  By Lemma \ref{lemma:headers-parent}, $w$ is $u$ or a sibling of $u$ in $T$. Thus $y'$ is $t(u)$ or a sibling of $u$ and $v$ in $T$.  If $y' = t(u)$, $v$ does not dominate $u$.
If $y'$ is a sibling of $u$, then $v$ does not dominate $y'$ by the choice of $u$, which implies that $v$ does not dominate $u$.
Thus in any case $v$ does not dominate $u$, so the sibling property holds.
\end{proof}

Algorithm 6, our low-high ordering algorithm using a loop nesting forest, is like Algorithm 4; it builds for each vertex an ordered list of its children in the dominator tree $D$ and then does a depth-first traversal of $D$ to find a low-high order. It inserts the vertices in increasing order into the lists of children in a way that gives the low-high property. As part of the insertion process, it defines, for each vertex $u \not= s$, a derived arc called the \emph{pivot arc} of $u$.  If the pivot arc of $u$ is not $(d(u), u)$, it also defines another derived arc called the \emph{test arc} of $u$.  Of the pivot and test arcs for $u$, one is $(f'(u), u)$, and the other is the derived arc $(y', w)$ of an arc $(y, w)$ satisfying Lemma \ref{lemma:headers-sibling}.  The test arc for $u$ is defined just before $u$ is inserted into a list of children.  The pivot arc for $u$ is defined either just before $u$ is inserted or earlier, when a cycle containing $u$ is (implicitly) expanded.  The pivot and test arcs determine where to insert $u$ in its list of siblings. An example is shown in Figure \ref{fig:headers-3}.

\begin{figure}
\begin{center}
\fbox{
\begin{minipage}[h]{16cm}
\begin{center}
\textbf{Algorithm 6: Construction of a Low-High Order using a Loop Nesting Forest}
\end{center}
Let $F$ be a depth-first spanning tree of $G$, with vertices numbered in reverse postorder and identified by number. Let $H$ be the loop nesting forest defined by $F$, with associated outgoing arcs and corresponding original arcs.
\begin{description}\setlength{\leftmargin}{10pt}
\item[Step 1:] Construct the derived graph $G'$ with respect to the dominator tree $D$.
\item[Step 2:] For each vertex $u$, initialize its list of children $C(u)$ to be empty. Apply the following steps to each vertex $u \not= s$ in increasing order:
    \begin{description}\setlength{\leftmargin}{10pt}
    \item[Step 2a:] Let $(x', v)$ be the pivot arc of $u$; if $u$ does not yet have a pivot arc, let $(x', v) = (f'(u), u)$.
    \item[Step 2b:] If $(x', v) = (d(u), u)$, insert $u$ first on $C(d(u))$, completing  Step 2. Otherwise, proceed to Step 2c.
    \item[Step 2c:] If $(x', v) \not= (f'(u), u)$, let the test arc $(y', w)$ of $u$ be $(f'(u), u)$; otherwise, let the test arc of $u$ be  the derived arc $(y', w)$ of an arc $(y, w)$ with $y < u$, $w$ a descendant of $u$ in $H$, and $y' \not= f'(u)$.
    \item[Step 2d:] If $x' = d(u)$, insert $u$ first in $C(d(u))$.  Otherwise, insert $u$ just before $x'$ in $C(d(u))$ if $y' = d(u)$ or $y'$ precedes $x'$ in $C(d(u))$, just after $x'$ otherwise.  Swap $(x', v)$ and $(y', w)$ if necessary so that $(y', w) = (f'(u), u)$.  If $v = u$, this completes Step 2; otherwise, proceed to Step 2e.
    \item[Step 2e:] Let $z$ be the child of $u$ in $H$ that is an ancestor of $v$.  Make $(x', v)$ be the pivot arc of $z$. Find a path $P$ of vertices $z=x_0,x_1,\ldots,x_k=u$ in the interval of $u$ by starting at $z$ and following outgoing arcs in the interval until reaching $u$. For each vertex $x_i$ other than $z$ and $u$ on $P$, find the original arc $(p, q)$ corresponding to the arc entering $x_i$ on $P$, and let the pivot arc of $x_i$ be the derived arc $(p', q)$ of $(p, q)$.
    \end{description}
\item[Step 3:] Do a depth-first traversal of $D$, visiting the children of each vertex $v$ in their order in $C(v)$. Number the vertices from $1$ to $n$ as they are visited. The resulting order is low-high on $G$.
\end{description}
\end{minipage}
}
\end{center}
\end{figure}

\begin{figure}
\begin{center}
\scalebox{.65}[.65]{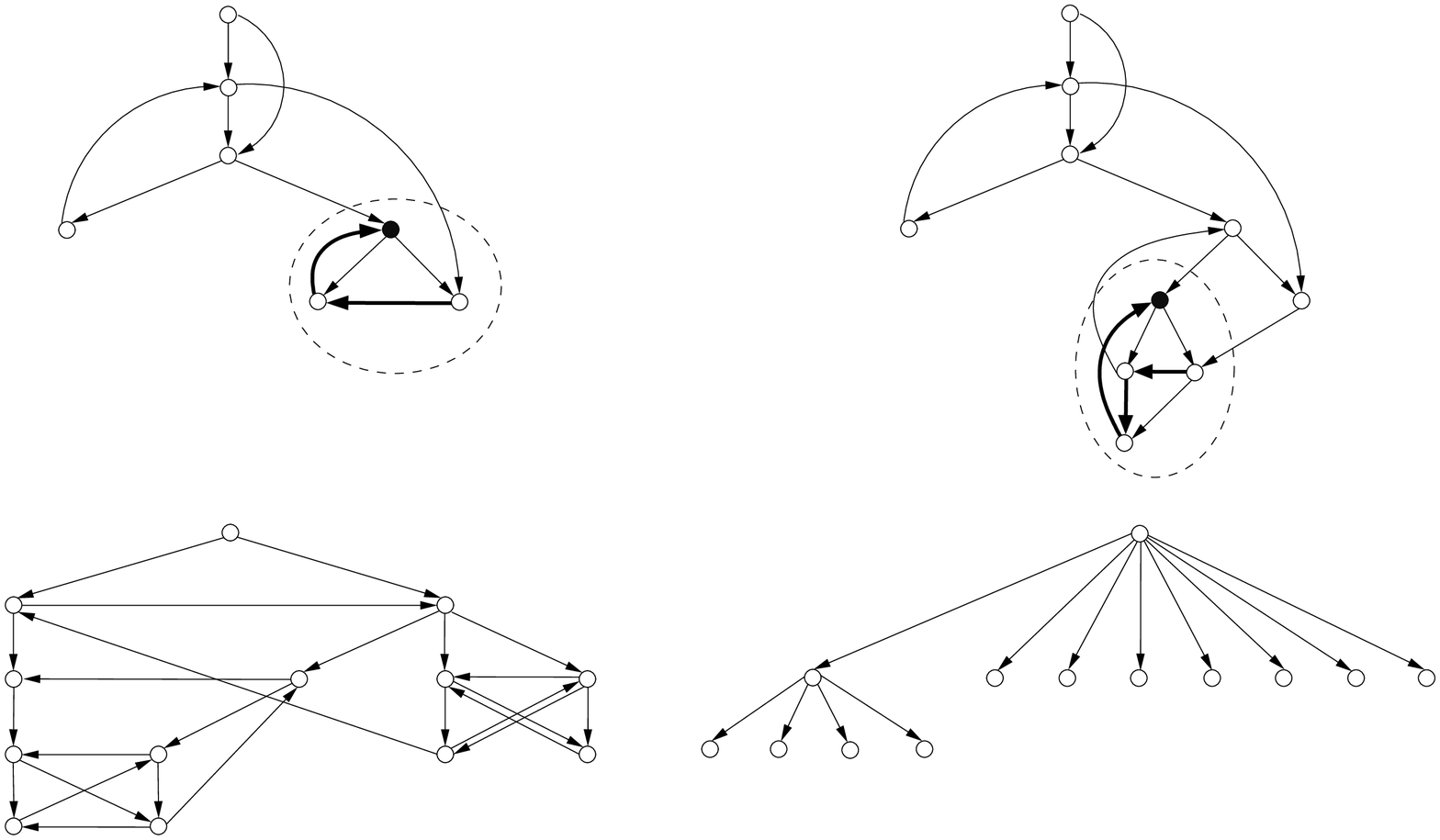}
\end{center}
\caption{Computation of a low-high order (inside the brackets) of the flow graph of Figure \ref{fig:headers}, using the intervals of Figure \ref{fig:headers-2}. When vertex $d$ is processed, Algorithm 6 sets $(b,d)$ as the pivot arc in Step 2a and $(a,c)$ as the test arc in Step 2c. In Step 2e it finds the path $P=(c,h,d)$ (shown with bold arcs) inside the interval of $d$ and sets $(c,g)$ as the pivot arc of $h$, which is the original arc corresponding to $(c,h)$. When vertex $h$ is processed, Algorithm 6 sets $(d,h)$ as the test arc in Step 2c. In Step 2e it finds the path $P=(g,l,k,h)$ (shown with bold arcs) inside the interval of $h$ and sets $(g,l)$ and $(l,k)$ as the pivot arcs of $l$ and $k$ respectively.\label{fig:headers-3}}
\end{figure}

\begin{theorem}
\label{theorem:low-high-headers}
Algorithm 6 is correct.
\end{theorem}
\begin{proof}
Consider the choice of a pivot arc $(x', v)$ for a vertex $u \not= s$.  We claim that $(x', v)$ is the derived arc of an original arc $(x, v)$ such that $x < u$ and $v$ is a descendant of $u$ in $H$.  This is immediate if $(x', v)$ is chosen in Step 2a. In  Step 2e, $v$ is a descendant of $z$ in $H$ and $x < u < z$, so the choice of $(x', v)$ as the pivot arc for $z$ satisfies the claim. Also in Step 2e, each original arc $(p, q)$ corresponding to an arc on $P$ entering a vertex $x_i$ is such that $q$ is a descendant of $x_i$ in $H$ but $p$ is not. Lemma \ref{lemma:headers} implies $p<x_i$, so the choice of $(p', q)$ as the pivot arc for $x_i$ satisfies the claim.

If $(x', v) = (d(u), u)$ in Step 2b, $u$ will satisfy the low-high property. Consider the processing of a vertex $u$ such that its pivot arc $(x', v) \not= (d(u), u)$.  If $(x', v) = (f'(u), u)$, $f(u) \not= d(u)$, so there is an arc $(y, w)$ satisfying Lemma \ref{lemma:headers-sibling}, and the algorithm will successfully choose the derived arc $(y', w)$ of such an arc as the test arc for $u$.  One of $v$ and $w$ is $u$; the other is $u$ or a sibling of $u$ in $D$ by Lemma \ref{lemma:headers-parent}.  Thus each of $x'$ and $y'$ is $d(u)$ or a sibling of $u$ in $D$.  Furthermore both are less than $u$, so they have already been inserted into lists of children when $u$ is about to be inserted.  Hence the insertion of $u$ is well defined.  If $v = w = u$, the insertion position of $u$ guarantees that $u$ satisfies the low-high property.

It remains to show that $u$ satisfies the low-high property even if one of $v$ and $w$ is not $u$. The choice of additional pivot arcs in Step 2e is what guarantees this. If $f(u) = d(u)$, $u$ will satisfy the low-high property, since then it does not matter where $u$ is inserted into $C(d(u))$.  Thus suppose $f(u) \not= d(u)$.  Then $f'(u)$ is a sibling of $u$ in $D$. Consider the arc $(x', v)$ in Step 2e (after the conditional swap in Step 2d makes $(x', v) \not= (f'(u), u)$).  Since $v \not= u$, $v$ is a sibling of $u$ in $D$, and $z$ is $v$ or a sibling of $v$ in $D$, so $z$ is a sibling of $u$ in $D$. Arc $(x', v)$ becomes the pivot arc of $z$.

We claim that after all insertions into $C(d(u))$, $u$ is between $f'(u)$ and $z$.  To prove the claim, we consider three cases.  If $x' = d(u)$ and $(x', v)$ is the pivot arc of $u$, then $u$ will be inserted first in $C(d(u))$, in front of $f'(u)$, and later $z$ will be inserted first, so $u$ is between $f'(u)$ and $z$.  If $x' = d(u)$ and $(x', v)$ is the test arc of $u$, then $(f'(u), u)$ is the pivot arc of $u$, so $u$ will be inserted in front of $f'(u)$ and later $z$ will be inserted first, so again $u$ is between $f'(u)$ and $z$.  The third case is $x' \not= d(u)$.  Then $u$ will be inserted between $f'(u)$ and $x'$.  Since $(x', v)$ is the pivot arc of $z$, $z$ will eventually be inserted next to $x'$, leaving $u$ between $f'(u)$ and $z$ in this case as well.

Now consider the vertices $z = x_0, x_1, \ldots, x_k = u$ on $P$. Let $S_i$ be the set of descendants of $x_i$ in $H$, and let $S$ be the union of the $S_i$. Let $(x, v)$ be the original arc whose derived arc is $(x', v)$.  There are $u$-avoiding paths from $s$ to all vertices in $S$, via the path in $F$ from $s$ to $x'$ followed by $(x', v)$ followed by a path among vertices in $S$. Thus $u$ dominates no vertex in $S$.  But $d(u)$ must dominate all vertices in $S$, since otherwise there would be a path from $s$ to $u$ avoiding $d(u)$.  Thus all vertices in $S$ are descendants of $d(u)$ in $D$.

We claim that $S$ consists of a set of siblings of $u$ in $D$ and possibly some of their descendants. This is true of $S_0$, since by Lemma \ref{lemma:parent-strongly-connected} $S_0$ consists of a set of siblings in $D$ and possibly some of their descendants, $x_0$ must be one of these siblings since it is minimum in $S_0$, and $z = x_0$ is a sibling of $u$.  Suppose the claim is false.  Let $x_i$ be minimum such that some vertex in $S_i$ is a descendant of a sibling of $u$ in $D$ but that sibling is not in $S$.  By Lemma \ref{lemma:parent-strongly-connected}, $S_i$ consists of a set of siblings in $D$ and possibly some of their descendants.  Let $(p, q)$ be the original arc corresponding to the arc entering $x_i$ on $P$. Then $p$ is in $S_{i-1}$.  By the choice of $x_i$, $p$ is a descendant in $D$ of a sibling $b$ of $u$ that is in $S$. Either $q$ is a descendant of $b$ in $D$, in which case all vertices in $S_i$ are descendants of $b$ in $D$, or $q$ is not a descendant of $b$.  In the latter case, the parent property implies that $q$ is a sibling of $b$ and hence of $u$, and $S_i$ consists of a set of siblings of $b$ and hence of $u$, and possibly some of their descendants. We conclude that there is no such $x_i$, verifying the claim.

Finally we claim that all of the siblings of $u$ in $S$ are on the same side of $u$ in $C(d(u))$ as $z = x_0$.  This implies that $u$ has the low-high property, because $u$ will be between $f'(u)$ and $p'$, where $(p', u)$ is the derived arc of the original arc $(p, u)$ corresponding to the arc entering $u$ on $P$. Suppose the claim is false for some sibling of $u$ in $S$. Let $b$ be the minimum such sibling. Let the pivot arc of $b$ be the derived arc $(a', c)$ of original arc $(a, c)$. Either $(a', c) = (x', v)$, or $a'$  and $c$ are in $S$.  In the former case, if $x' = d(u)$, $b$ will be inserted at the front of $C(d(u))$ sometime after $z$ was inserted at the front, so $b$ will be on the same side of $u$ as $z$; if $x'$ is a sibling of $u$, $b$ will be inserted next to $x'$ sometime after $z$ was inserted next to $x'$, so again $b$ will be on the same side of $u$ as $z$. In the latter case, by Lemma \ref{lemma:headers-parent} $c$ is a sibling of $b$ and hence of $u$, and since $a'$ is in $S$, $a'$ is also a sibling of $u$.  Furthermore $a'<b$, so $a'$ is on the same side of $u$ as $z$ by the choice of $b$.  The insertion of $b$ next to $a'$ puts it on the same side as well.
\end{proof}

\begin{theorem}
\label{theorem:low-high-headers-2}
A tree with the parent property has the sibling property, and hence is the dominator tree, if and only if it has a low-high order with respect to $G$.
\end{theorem}
\begin{proof}
Immediate from Theorems \ref{theorem:parent-low-high} and \ref{theorem:low-high-headers}.
\end{proof}

Excluding the computation of the loop nesting forest and associated information, it is straightforward to implement Algorithm 6 to run in $O(m)$ time.  To find test arcs in Step 2c, we compute, for each vertex $u \not= s$, two derived  arcs $(x', v)$ and $(y', w)$ with $v$ and $w$ descendants of $u$ in $H$, $x' \not= y'$, and $x'$ and $y'$ minimum, if two such arcs exist; otherwise, $f(u) = d(u)$.  We can do this by processing the vertices of $H$ bottom-up (from leaves to roots). This takes $O(m)$ time. We do the insertions into the lists of children using the off-line method developed in Section \ref{sec:list-order}.

We can if we wish run Algorithm 6 on the derived graph instead of the original graph. This is appealing if the loop nesting forest is not given but must be computed, since it simplifies the structure of the strongly connected subgraphs and simplifies Step 2. In the derived graph, every arc leads from a vertex to a child or sibling in $D$, so the vertex set of every strongly connected subgraph consists of a set of siblings in $D$. Furthermore every such subgraph with at least two vertices has at least two entry vertices from $s$. To run Algorithm 6 on the derived graph, we begin by computing the derived graph $G'$. Then we do a depth-first search of $G'$ to generate a depth-first spanning tree $F$ and to number the vertices in reverse preorder. Next we construct the loop nesting forest of $G'$ defined by $F$, along with associated in-trees and corresponding original arcs. Finally, we run Steps 2 and 3, ignoring the distinction between original arcs and derived arcs: in $G'$, every arc is its own derived arc.
\\

\begin{remark}
If we want to verify the dominator tree of an arbitrary graph but are not interested in constructing a low-high order, we can do the verification using Lemma \ref{lemma:headers-sibling}. Given $D$, we verify that $D$ is a tree with the parent property as in Section \ref{sec:properties}, compute the derived graph as in Section \ref{sec:derived-graph}, (verifying the correctness of the derived graph using ancestor-descendant tests), construct a loop nesting forest $H$ (for either the original or derived graph) and verify the condition in Lemma \ref{lemma:headers-sibling}, which can be done by processing the vertices of $H$ from leaves to roots.  The total time for verification by this method is $O(m)$.  This method assumes that the loop nesting forest is correct.  To verify that it is, we verify for each $u \not= s$ that all its children in $H$ are descendants of $u$ in $F$, that the interval of $u$ is strongly connected, and that it becomes acyclic when $u$ is deleted. This also takes $O(m)$ time.
\end{remark}

\section{Low-High Orders from Semi-Dominators}
\label{sec:semidominators}

\subsection{Low-high orders from independent spanning trees}
\label{sec:low-high-from-ist}

In this section we develop an alternative algorithm for finding a low-high order in an arbitrary graph.  Instead of using loop nesting information, the algorithm uses semi-dominator information.  This information is computed by the fast algorithms for finding dominators~\cite{domin:ahlt,dominators:bgkrtw,domin:lt}, making it easy to extend these algorithms to find not only the dominator tree, but a low-high order as well.
The algorithm has two steps.  The first step uses the semi-dominator information to build two independent spanning trees.  The second step uses two independent spanning trees to find a low-high order.  Since the two steps are independent and the second step requires no new ideas and is much easier to prove correct, we begin with it.

Let $B$ and $R$ be two independent spanning trees: for each vertex $v$, the paths in $B$ and $R$ from $s$ to $v$ have only the dominators of $v$ in common.  We denote by $b(v)$ and $r(v)$ the parent of $v$ in $B$ and $R$, respectively.  We begin by slightly modifying $B$ and $R$: for each vertex $v \not= s$, if $(d(v), v) \in A$ we replace $b(v)$ and $r(v)$ by $d(v)$.  Then we find the derived arcs of the arcs in $B$ and $R$.  Let $G'$ be the subgraph whose vertex set is $V$ and whose arcs are the derived arcs of those in $B$ and $R$.  Let $B'$ and $R'$ be the subgraphs of $G'$ defined by the derived arcs of $B$ and $R$, respectively.
See Figure \ref{fig:low-high-ist}.

\begin{figure}
\begin{center}
\scalebox{0.7}[0.7]{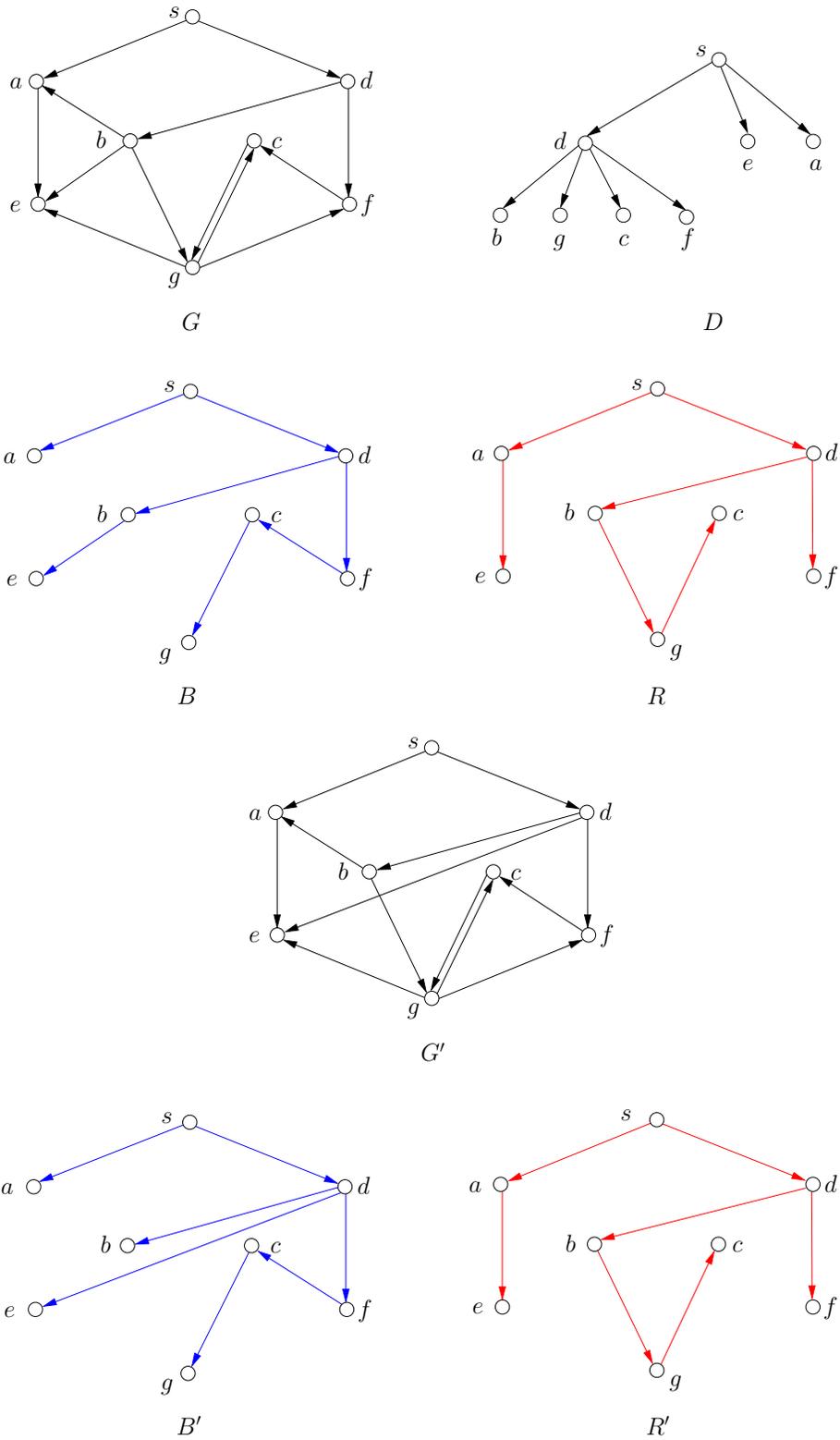}
\end{center}
\caption{\label{fig:low-high-ist} A flow graph, its dominator tree (with vertices arranged in a low-high order), and two independent spanning trees; the derived graph and two independent spanning trees.}
\end{figure}

\begin{lemma}
\label{lemma:ist-derived}
Subgraphs $B'$ and $R'$ are independent spanning trees in $G'$.
\end{lemma}
\begin{proof}
Any path from $s$ to $v$ in $B$ or $R$ contains $d(v)$, so if $(d(v), v) \in A$, replacing $b(v)$ and $r(v)$ by $d(v)$ leaves $B$ and $R$ trees.  Furthermore such replacements only eliminate vertices from paths in $B$ and $R$, so $B$ and $R$ remain independent.  Similarly, for any vertex $v \not= s$, $v$ does not dominate $b(v)$ or $r(v)$, so every arc in $B$ and $R$ has a derived arc.  If $v \not= s$, $b'(v)$ ($r'(v)$) is on any path from $s$ to $b(v)$ ($r(v)$), and hence on the path from $s$ to $v$ in $B$ ($R$). It follows that replacing all the arcs in $B$ and $R$ by their derived arcs produces two independent spanning trees in $G'$.
\end{proof}

As in Algorithms 4 and 6, the tree-based low-high ordering algorithm, Algorithm 7, builds ordered lists of the children in $D$ of each vertex $v \not= s$.  It uses the arcs in $B'$ and $R'$ to determine the insertion positions.  It inserts vertices in an order determined recursively, unlike the iterative orders used in Algorithms 4 and 6 (topological order in Algorithm 4, reverse postorder in Algorithm 6).

\begin{figure}
\begin{center}
\fbox{
\begin{minipage}[h]{16cm}
\begin{center}
\textbf{Algorithm 7: Construction of a Low-High Order Given Two Independent Spanning Trees}
\end{center}
Let $B$ and $R$ be independent spanning trees.
\begin{description}\setlength{\leftmargin}{10pt}
\item[Step 1:] For each vertex $v \not= s$, if $(d(v), v) \in A$, replace $b(v)$ and $r(v)$ by $d(v)$.
\item[Step 2:] Compute the derived arcs of the arcs in $B$ and $R$, and let the resulting derived graph be $G'$. Let $B'$ and $R'$, respectively, be the spanning trees in $G'$ whose arcs are the derived arcs of $B$ and $R$, respectively.
\item[Step 3:] For each $v$, initialize its list of children $C(v)$ to be empty.
\item[Step 4:] If $G'$ contains only one vertex $v \not= s$, insert $v$ anywhere in $C(s)$. Otherwise, let $v$ be a vertex whose in-degree in $G'$ exceeds its number of children in $B'$ plus its number of children in $R'$. Assume $v$ is a leaf in $R'$; proceed symmetrically if $v$ is a leaf in $B'$.  If $v$ is not a leaf in $B'$, let $w$ be its child in $B'$, and replace $b'(w)$ by $b'(v)$. Delete $v$, and apply Step 4 recursively to insert the remaining vertices other than $s$ into lists of children. If $b'(v) = d(v)$, insert $v$ anywhere in $C(d(v))$; otherwise, insert $v$ just before $b'(v)$ in $C(d(v))$ if $r'(v)$ is before $b'(v)$ in $C(d(v))$, just after $b'(v)$ otherwise.
\item[Step 5:] Do a depth-first traversal of $D$, visiting the children of each vertex $v$ in their order in $C(v)$. Number the vertices from $1$ to $n$ as they are visited. The resulting order is low-high on $G$.
\end{description}
\end{minipage}
}
\end{center}
\end{figure}

Figure \ref{fig:low-high-ist-2} illustrates how this algorithm works.

\begin{figure}
\begin{center}
\scalebox{0.62}[0.62]{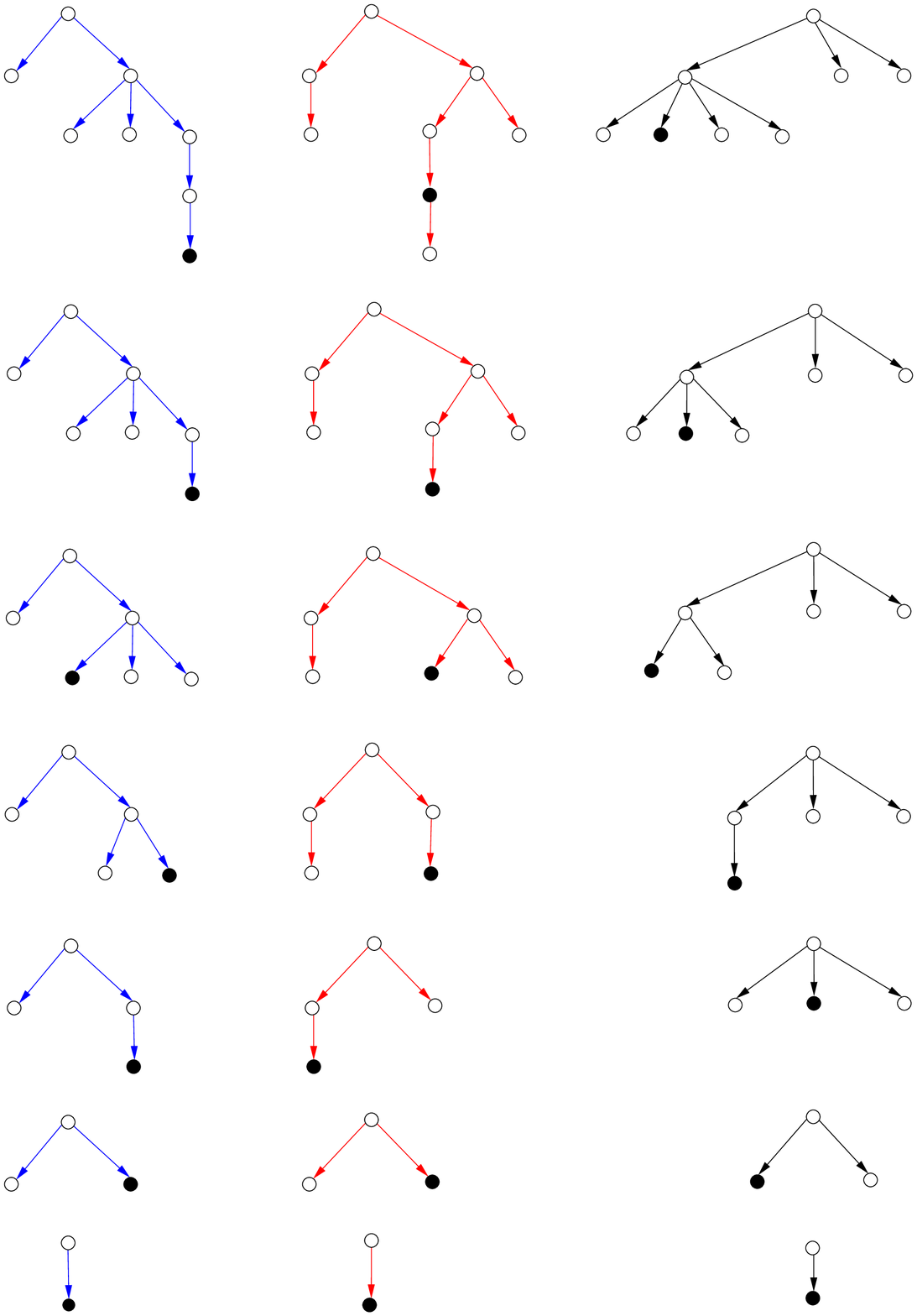}
\end{center}
\caption{\label{fig:low-high-ist-2} Computation of a low-high order starting from two independent spanning trees of the derived graph. The vertex removed at each application of Step 4 is shown filled.}
\end{figure}

\begin{theorem}
The vertex order computed by Algorithm 7 is low-high on $G'$ and hence on $G$.
\end{theorem}
\begin{proof}
We claim that Step 4 maintains the invariant that (i) $B'$ and $R'$ are independent spanning trees rooted at $s$; (ii) for every $v \not= s$, either $b'(v) = r'(v) = d(v)$, or $b'(v)$, $r'(v)$, and $d(v)$ are all distinct; and (iii) each arc in $B'$ or $R'$ is its own derived arc.
Let $u$ be a non-leaf of $D$ all of whose children in $D$ are leaves of $D$.  Let $X$ be the set of children of $u$ in $D$, and let $Y$ be the subset of $X$ that consists of the vertices $y$ such that $b'(y)=r'(y)=d(y)=u$. The invariant implies that each vertex in $Y$ has in-degree 1 in $G'$ and each vertex in $X - Y$ has in-degree 2 in $G'$.  Since each arc in $B'$ or $R'$ is its own derived arc, each arc in $G'$ leaving a vertex in $X$ enters a vertex in $X$. Since at least one arc from $u$ enters $X$, there must be a vertex $v$ in $X$ whose in-degree in $G'$ exceeds its out-degree in $G'$.  We claim that $v$ can be selected in Step 4.  Indeed, $v$ can be selected if is a leaf in both $B'$ and $R'$.  If not, then $v$ must have in-degree 2 and out-degree 1 in $G'$.  But then $v$ must be a leaf in either $B'$ or $R'$; if not, its common child $w$ in $B'$ and $R'$ would violate the invariant; $B'$ and $R'$ would not be independent spanning trees, since $v$ is an ancestor of $w$ in both $B'$ and $R'$, but $v$ and $w$ are siblings in $D$.  Hence $v$ can be selected in this case also. It follows that replacing $b'(w)$ by $b'(v)$ and deleting $v$ preserves the invariant.
\ignore{
The claim is true before any vertices are deleted.  Let $v$ be a vertex chosen for deletion.  If $v$ has no outgoing arcs, deletion of $v$ preserves the invariant, since $v$ dominates no vertices.  The choice of $v$ guarantees that if $b'(v) = r'(v) = d(v)$, $v$ has no outgoing arcs.  Thus suppose $b'(v)$, $r'(v)$, and $d(v)$ are distinct, and that $b'(w) = v$; the argument is symmetric if $r'(w) = v$. The choice of $v$ guarantees that $r'(w) \not= v$, which implies that $v$ dominates no vertices.  In particular, $v \not= d(w)$, so $b'(v)$, $v$, and $w$ are siblings in $D$.  It cannot be the case that $r'(w) = b'(v)$, for then the paths in $B'$ and $R'$ from $s$ to $w$ would share $b'(v)$, which does not dominate $w$.  Thus $r'(w)$, $v$, and $b'(v)$ are distinct siblings in $D$.  It follows that replacing $b'(w)$ by $b'(v)$ and deleting $v$ preserves the invariant.

Next we claim that Step 4 runs to completion; that is, there is always a vertex $v$ to choose in Step 4.  In $G'$ the total in-degree equals the total out-degree.  Since $s$ has in-degree zero and positive out-degree, at least one vertex $v \not= s$ has in-degree exceeding its out-degree and hence is a candidate for $v$.  Since $v$ has in-degree one or two, its out-degree is zero or one, which means it is a leaf in $B'$ or $R'$.
}

Finally, we claim that the computed order is low-high.  This is immediate if $G'$ has two vertices.  Suppose this is true if $G'$ has $k \ge 2$ vertices.  Let $G'$ have $k + 1$ vertices and let $v$ be the vertex chosen for deletion.  The insertion position of $v$ guarantees that $v$ has the low-high property.  All vertices in $G'$ after the deletion of $v$ have the low-high property in the new $G'$ by the induction hypothesis, so they have the low-high property in the old $G'$ with the possible exception of $w$, one of whose incoming arcs differs in the old and the new $G'$.  Suppose $b'(w)$ differs; the argument is symmetric if $r'(w)$ differs.  By the argument above, $v$, $w$, $b'(v)$, and $r'(w)$ are distinct siblings in $D$.  Since $w$ has the low-high property in the new $G'$, it occurs in $C(d(v))$ between $r'(w)$ and $b'(v)$.  Insertion of $v$ next to $b'(v)$ leaves $w$ between $r'(w)$ and $v$, so it has the low-high property in the old $G'$ as well.
\end{proof}

Using Algorithm 5 to implement Step 4, it is straightforward to implement Algorithm 7 to run in $O(n)$ time.

\subsection{Independent spanning trees from semi-dominators}
\label{sec:two-independent-spanning-trees}

Now we turn to the problem of finding two independent spanning trees on which to run Algorithm 7.  We need to define semi-dominators and related concepts.  These are defined with respect to a vertex numbering given by depth-first search, but in contrast to the reverse postorder numbering used in Section \ref{sec:low-high-reducible}, this numbering is preorder, the order in which the vertices are first visited by the search.  The tree construction algorithm does not actually use any vertex numbering; we introduce the preorder numbering just for definitions and analysis.

Let $F$ be a depth-first spanning tree of $G$ rooted at $s$, with vertices numbered from $1$ to $n$ as they are first visited by the search.  Identify vertices by number.  We shall use the following basic lemma about depth-first search repeatedly:

\begin{lemma} \emph{(Path Lemma \cite{dfs:t})}
\label{lemma:dfs}
If $v$ and $w$ are vertices such that $v<w$, then any path from  $v$ to $w$ contains a common
ancestor of $v$ and $w$ in $F$.
\end{lemma}

A path from $u$ to $v$ is \emph{high} if all its vertices other than $u$ and $v$ are higher than both $u$ and $v$.  If $v \not= s$, the \emph{semi-dominator} $\mathit{sd}(v)$ is the minimum vertex $u$ such that there is a high path from $u$ to $v$.  Since $f(v)$ is a candidate, $\mathit{sd}(v) \le f(v)$.  By Lemma \ref{lemma:dfs}, $\mathit{sd}(v)$ is a proper ancestor of $v$ in $F$.  Indeed, $\mathit{sd}(v)$ is the ancestor $u$ of $v$ in $F$ closest to $s$ such that there is a path from $u$ to $v$ avoiding all other vertices on $F[u, v]$ (the path in $F$ from $u$ to $v$).  Since there is a path from $s$ to $v$ avoiding all vertices on $F[\mathit{sd}(v), v]$ except $\mathit{sd}(v)$ and $v$, $d(v) \le \mathit{sd}(v)$.  The \emph{relative dominator} $\mathit{rd}(v)$ is the vertex $x$ on $F(\mathit{sd}(v), v]$ (the path in $F$ from a child of $\mathit{sd}(v)$ to $v$) such that $\mathit{sd}(x)$ is minimum, with a tie broken in favor of the smallest $x$.

Semi-dominators and relative dominators provide a simple way to compute dominators:

\begin{lemma} \emph{(\cite{domin:lt})}
\label{lemma:semi-relative-dominators}
Let $v \not= s$. If $\mathit{sd}(v) = \mathit{sd}(\mathit{rd}(v))$, then $d(v) = \mathit{sd}(v)$; otherwise, $d(v) = d(\mathit{rd}(v))$.
\end{lemma}

By applying the recurrence in Lemma \ref{lemma:semi-relative-dominators} to the vertices in increasing order, one can find the dominators in $O(n)$ time.  This is the last step of the Lengauer-Tarjan~\cite{domin:lt} and Buchsbaum et al.~\cite{dominators:bgkrtw} algorithms for finding dominators.  The former uses path compression to compute semi-dominators and relative dominators; the latter uses path compression and additional techniques.

Semi-dominators also provide a necessary and sufficient condition for a tree with the parent property to have the sibling property:

\begin{lemma}
\label{lemma:semidominator-sibling}
A tree $T$ with the parent property has the sibling property, and hence is $D$, if and only if, for every $v \not= s$, $f(v) = d(v)$ or $\mathit{sd}(v) < f'(v)$, where $(f'(v), v)$ is the derived arc of $(f(v), v)$ with respect to $T$.
\end{lemma}
\begin{proof}
Suppose $T$ has the sibling property.  Then $T = D$ by Theorem \ref{theorem:parent-sibling-2}.  Let $v \not= s$ be a vertex such that $f(v) \not= d(v)$.  Then $f'(v)$ is a sibling of $v$ in $T$.  By the sibling property, there is an $f'(v)$-avoiding simple path $P$ from $s$ to $v$.  Let $u$ be the last vertex on $P$ less than $v$.  The part of $P$ from $u$ to $v$ is a high path, so $\mathit{sd}(v) \le u$.  By Lemma \ref{lemma:dfs}, $u$ is a proper ancestor of $v$ in $F$ and hence an ancestor of $f(v)$.  Suppose $u > f'(v)$.  Then there is a path from $s$ to $f(v)$ that avoids $f'(v)$, consisting of the part of $P$ from $s$ to $u$ followed by $F[u, f(v)]$.  But $f'(v)$ dominates $f(v)$, so there is no such path.  Thus it must be the case that $u < f'(v)$, which implies $\mathit{sd}(v) < f'(v)$.  Hence the condition in the lemma holds.

To prove the converse, suppose that $T$ does not have the sibling property.  Choose $u$ and $v$ with $v$ minimum such that $u$ and $v$ are siblings in $T$ and $u$ dominates $v$.  Then $u$ is an ancestor of $f(v)$ in $F$.  Furthermore $u$ dominates $f(v)$; otherwise, it would not dominate $v$.  By the parent property, $f(v)$ is not a descendant of $v$ in $T$, for then $v$ would dominate $f(v)$.  It is not a proper ancestor of $u$ in $T$, for then $f(v)$ would dominate $u$.  Nor can it be a descendant in $T$ of a sibling $x$ of $u$ and $v$, for then $x \le f(v) < v$ and $x$ dominates $f(v)$, which implies that one of $x$ and $u$ dominates the other, contradicting the choice of $v$.  The only remaining possibility is that $f(v)$ is a descendant of $u$ in $T$, possibly $u$ itself.  This implies $f'(v) = u$.  It cannot be the case that $\mathit{sd}(v) < u$, for then $F[s, \mathit{sd}(v)]$ would avoid $u$, as would the high path from $\mathit{sd}(v)$ to $v$, so $u$ would not dominate $v$.  Hence $\mathit{sd}(v) \ge u$, and $v$ violates the condition in the lemma.
\end{proof}

If the semi-dominators are available, we can use Lemma \ref{lemma:semidominator-sibling} to verify that a tree with the parent property has the sibling property and hence is $D$. But we know of no simple $O(m)$-time way to verify the correctness of the semi-dominators.  We thus prefer to use a low-high order to verify the sibling property, since it gives a complete solution to the dominator verification problem.

Our algorithm to construct two independent spanning trees requires as input a depth-first spanning tree $F$ along with the corresponding semi-dominators, relative dominators, and, for each vertex $v \not= s$, an arc $(g(v), v)$ that is last on some high path from $\mathit{sd}(v)$ to $v$.  It may happen that $g(v) = f(v)$, but only if $f(v) = d(v)$.  The converse need not be true.

It is straightforward to augment the computation of semi-dominators to compute last arcs as well; indeed, such a computation is implicit in the computation of semi-dominators.  The increase in running time is at most a constant factor.

To prove that the tree-construction algorithm is correct, we need the following technical lemma about semi-dominators:

\begin{lemma}
\label{lemma:semidominator-equal-node}
Let $v \not= s$ be a vertex such that $g(v) > v$, and let $a$ be the nearest common ancestor of $v$ and $g(v)$.  Then there is a vertex $w$ on $F(a, g(v)]$ such that $\mathit{sd}(w) = \mathit{sd}(v)$.
\end{lemma}
\begin{proof}
Let $P$ be a high path from $\mathit{sd}(v)$ via $g(v)$ to $v$, and let $w$ be the smallest vertex on this path other than $\mathit{sd}(v)$ and $v$.  We prove that $w$ satisfies the lemma.  Since $w > v \ge a$, $w$ is not an ancestor of $a$ in $F$.  The part of $P$ from $\mathit{sd}(v)$ to $w$ is a high path, so $\mathit{sd}(w) \le \mathit{sd}(v)$.  Any high path to $w$ can be extended to a high path to $v$ by adding the part of $P$ from $w$ to $v$, so $\mathit{sd}(w) \ge sd(v)$.  Hence $\mathit{sd}(w) = \mathit{sd}(v)$.  If $w = g(v)$, $w$ satisfies the lemma.  Otherwise, $w < g(v)$, so by Lemma \ref{lemma:dfs} there is a common ancestor of $w$ and $g(v)$ in $F$ on the part of $P$ from $w$ to $g(v)$.  Since $w$ is smallest on this part of $P$, it must be this ancestor.  Thus $w$ satisfies the lemma.
\end{proof}

Here is our algorithm to construct two independent spanning trees:
\begin{figure}[h]
\begin{center}
\fbox{
\begin{minipage}[h]{16cm}
\begin{center}
\textbf{Algorithm 8: Construction of Two Independent Spanning Trees $B$ and $R$}
\end{center}
For each vertex $v \not= s$ in increasing order, choose one of $f(v)$ and $g(v)$ to be the parent $b(v)$ of $v$ in $B$ and the other to be the parent $r(v)$ of $v$ in $R$, as follows: if $\mathit{sd}(v) = \mathit{sd}(\mathit{rd}(v))$ or $b(\mathit{rd}(v)) = f(\mathit{rd}(v))$, set $b(v) = g(v)$ and $r(v) = f(v)$; otherwise, set $b(v) = f(v)$ and $r(v) = g(v)$.
\end{minipage}
}
\end{center}
\end{figure}

We call a vertex $v \not= s$ \emph{blue} if $b(v) = g(v)$ and \emph{red} otherwise.  An equivalent way to state Algorithm 8 is: color $v$ blue if $\mathit{sd}(v) = \mathit{sd}(\mathit{rd}(v))$ or $\mathit{rd}(v)$ is red; color $v$ red otherwise.  Vertex $s$ has no color.  Figure \ref{fig:ind-st-construction} gives an example of the construction.

\begin{figure}
\begin{center}
\scalebox{1}[1]{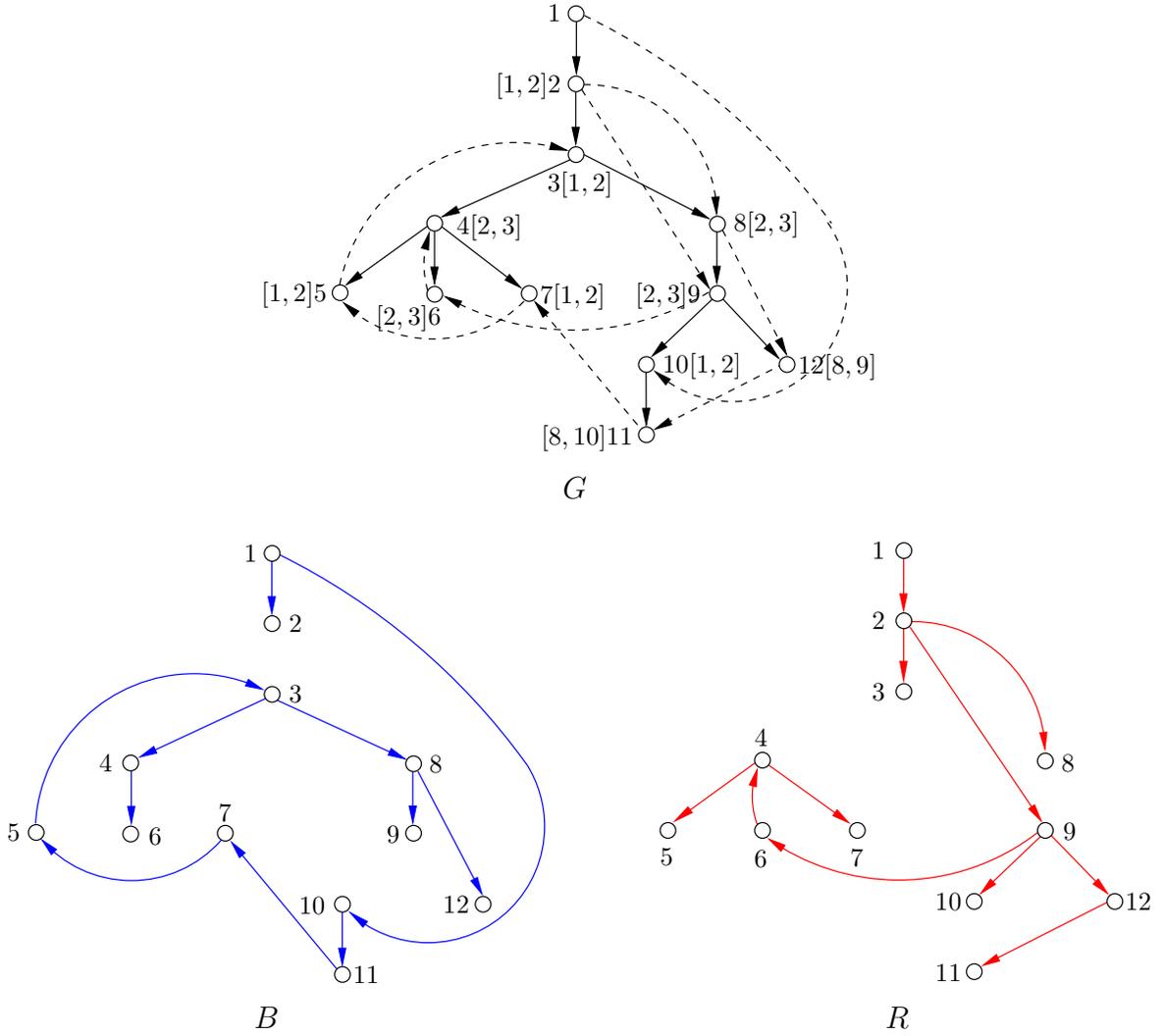}
\end{center}
\caption{Construction of two independent spanning trees $B$ and $R$. Vertices are numbered in preorder
with respect to a depth-first spanning tree $F$ shown with solid arcs; dashed arcs are not in
$F$. Numbers inside the brackets correspond to $\mathit{sd}(v)$ and $\mathit{rd}(v)$.
\label{fig:ind-st-construction}}
\end{figure}

Interestingly, Algorithm 8 does not use dominators at all, although we can simplify the coloring rule with their use, since $\mathit{sd}(\mathit{rd}(v)) = d(v)$: color $v$ blue if $\mathit{sd}(v) = d(v)$ or $\mathit{rd}(v)$ is red; color $v$ red otherwise.

Although Algorithm 8 is simple, its correctness proof is quite intricate. Indeed, it is amazing that the algorithm actually works. Its subtlety is suggested by the fact that there is a natural alternative way to define relative dominators such that Lemma \ref{lemma:semi-relative-dominators} holds but Algorithm 8 fails. Specifically, for $v \not= s$ let $\mathit{rd'}(v)$ be the largest vertex $x$ on $F(\mathit{sd}(v), v]$ with $sd(x) < sd(v)$ if such a vertex exists, $v$ otherwise.  Then Lemma \ref{lemma:semi-relative-dominators} holds with $\mathit{rd'}$ in place of $\mathit{rd}$, but the example in Figure \ref{fig:alt-rd} shows that Algorithm 8 fails with $\mathit{rd'}$ in place of $\mathit{rd}$.

\begin{figure}
\begin{center}
\scalebox{1.}[1.]{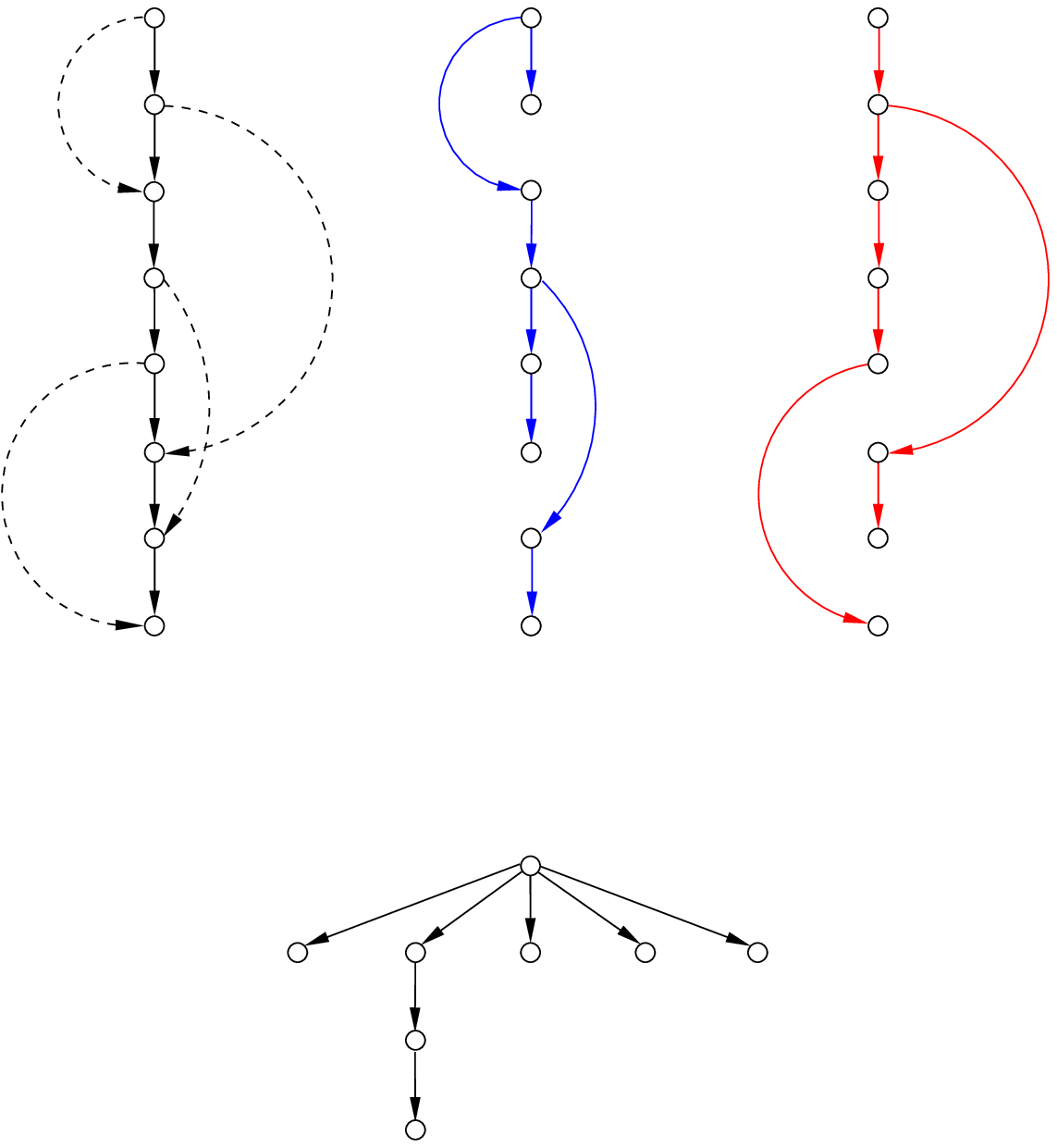}
\end{center}
\caption{An example where Algorithm 8 fails to compute two independent spanning trees if a natural alternative definition of relative dominators that satisfies Lemma \ref{lemma:semi-relative-dominators} is used.
Vertices are numbered in preorder
with respect to a depth-first spanning tree shown with solid arcs; dashed arcs are not in the tree. Numbers inside the brackets correspond to
$\mathit{sd}(v)$ and $\mathit{rd'}(v)$. Vertex $8$ violates the definition of independence.\label{fig:alt-rd}}
\end{figure}

To prove the correctness of Algorithm 8 (for the original definition of relative dominators), we first prove that $B$ and $R$ are acyclic and hence spanning trees rooted at $s$, and then that they are independent.  Both parts of the proof require some groundwork.  We begin with two lemmas that relate the colors of certain vertices.

\begin{lemma}
\label{lemma:2-vertices-color}
If $v$ and $w$ are vertices such that $v$ is an ancestor of $w$ in $F$, $\mathit{sd}(v) = \mathit{sd}(w)$, and $\mathit{sd}(x) \ge \mathit{sd}(v)$ for every $x$ on $F[v, w]$, then $v$ and $w$ are the same color.
\end{lemma}
\begin{proof}
The hypothesis of the lemma and the definition of $\mathit{rd}$ imply $\mathit{rd}(v) = \mathit{rd}(w)$.  The lemma follows from the equality $\mathit{sd}(v) = \mathit{sd}(w)$ and the coloring rule.
\end{proof}

Lemma \ref{lemma:2-vertices-color} implies that $B$ and $R$ do not change if we define $\mathit{rd}(v)$ to be \emph{any} vertex $x$ on $F(\mathit{sd}(v), v]$ with $\mathit{sd}(x)$ minimum, since by the lemma every such $x$ has the same color.  Choosing $x$ smallest simplifies some of our proofs, however.

\begin{lemma}
\label{lemma:3-vertices-color}
Let $x$, $y$, and $z$ be vertices other than $s$ such that
\begin{itemize}
\item[(i)]   $x$ is an ancestor of y and z in F and y and z are related in F;
\item[(ii)]	 $\mathit{sd}(x) < \mathit{sd}(y) < \mathit{sd}(z) < x$;
\item[(iii)] $\mathit{sd}(v) \ge \mathit{sd}(x)$ for all $v$ on $F[\min\{y, z\}, \max\{y, z\}]$
\item[(iv)]	 $x$ and $z$ are the same color.
\end{itemize}
Then $y$ is the same color as $x$ and $z$.
\end{lemma}
\begin{proof}
Suppose the lemma is false.  Let $x$, $y$, $z$ be three vertices that violate the lemma and such that $x$ is minimum.  We shall show that $\mathit{rd}(y) < x$ and $\mathit{rd}(y)$, $\mathit{rd}(z)$, $x$ is a triple violating the lemma, contradicting the choice of $x$, $y$, $z$.

By (i) and (ii), $x$ is a candidate for $\mathit{rd}(z)$, so $\mathit{sd}(\mathit{rd}(z)) \le \mathit{sd}(x)$.  Also by (ii), $\mathit{sd}(x) < \mathit{sd}(z)$, so $\mathit{sd}(\mathit{rd}(z)) < \mathit{sd}(z)$.  By the coloring rule, $\mathit{rd}(z)$ and $z$ have different colors.  By (iv), $\mathit{rd}(z)$ and $x$ have different colors.  If $\mathit{sd}(\mathit{rd}(z)) = \mathit{sd}(x)$, $\mathit{rd}(z)$ and $x$ would have the same color by Lemma \ref{lemma:2-vertices-color}; since they do not, $\mathit{sd}(\mathit{rd}(z)) < \mathit{sd}(x)$.  By (iii), $\mathit{rd}(z)$ is a proper ancestor of both $y$ and $z$ in $F$.  Thus $\mathit{rd}(z)$ is a candidate for $\mathit{rd}(y)$.  We claim $\mathit{sd}(\mathit{rd}(y)) \not= \mathit{sd}(\mathit{rd}(z))$.  Suppose by way of contradiction that $\mathit{sd}(\mathit{rd}(y)) = \mathit{sd}(\mathit{rd}(z))$.  By Lemma \ref{lemma:2-vertices-color}, $\mathit{rd}(y)$ and $\mathit{rd}(z)$ are the same color.  Since $\mathit{sd}(\mathit{rd}(z)) < \mathit{sd}(x)$, (ii) implies $\mathit{sd}(\mathit{rd}(y)) < \mathit{sd}(y)$ and $\mathit{sd}(\mathit{rd}(z)) < \mathit{sd}(z)$.  By the coloring rule, $y$ and $z$ are the same color, a contradiction.  Thus $\mathit{sd}(\mathit{rd}(y) \not= \mathit{sd}(\mathit{rd}(z))$.  Since $\mathit{rd}(z)$ is a candidate for $\mathit{rd}(y)$, $\mathit{sd}(\mathit{rd}(y)) < \mathit{sd}(\mathit{rd}(z))$ and $\mathit{rd}(y)$ is an ancestor of $\mathit{sd}(z)$ in $F$.

Since $\mathit{rd}(y)$ is an ancestor of $\mathit{sd}(z)$ in $F$, it is a proper ancestor of both $\mathit{rd}(z)$ and $x$.  In particular, $\mathit{rd}(y) < x$.  Also, $\mathit{rd}(x)$ and $x$ are both ancestors of $y$ (and $z$) in $F$, and hence related.  Thus (i) holds for $\mathit{rd}(y)$, $\mathit{rd}(z)$, $x$.   We have $\mathit{sd}(\mathit{rd}(y)) < \mathit{sd}(\mathit{rd}(z)) < \mathit{sd}(x) < \mathit{rd}(y)$, the last inequality following from (ii): $\mathit{sd}(x) < \mathit{sd}(y) < \mathit{rd}(y)$.  Thus (ii) holds for $\mathit{rd}(y)$, $\mathit{rd}(z)$, $x$.  Since $\mathit{rd}(z)$ and $x$ are proper descendants of $\mathit{sd}(y)$ and ancestors of $y$, (iii) holds for $\mathit{rd}(y)$, $\mathit{rd}(z)$, $x$.  Since $\mathit{sd}(\mathit{rd}(y)) < \mathit{sd}(\mathit{rd}(z)) < \mathit{sd}(x) < \mathit{sd}(y)$ (the last inequality by (ii)), the coloring gives $\mathit{rd}(y)$ and $y$ different colors.  Since $x$ and $y$ have different colors by assumption, $\mathit{rd}(y)$ and $x$ have the same color, so (iv) holds for $\mathit{rd}(y)$, $\mathit{rd}(z)$, $x$.  Finally, we showed above that $\mathit{rd}(z)$ and $x$ have different colors, so $\mathit{rd}(y)$, $\mathit{rd}(z)$, $x$ is a triple violating the lemma.
\end{proof}

We use Lemmas \ref{lemma:dfs}, \ref{lemma:semidominator-equal-node}, \ref{lemma:2-vertices-color}, and \ref{lemma:3-vertices-color} to prove that $B$ and $R$ are acyclic and hence spanning trees rooted at $s$.

\begin{theorem}
\label{theorem:semidominators-spanning-trees}
Both $B$ and $R$ are spanning trees rooted at $s$.
\end{theorem}
\begin{proof}
Suppose not.  Then $B$ or $R$ contains a cycle.  We shall construct a triple $x$, $y$, $z$ violating Lemma \ref{lemma:3-vertices-color}, yielding the theorem by contradiction.  Let $x$ be the minimum vertex on the cycle.  Then $x \not= s$ (since $s$ contains no incoming arcs), so $\mathit{sd}(x) < x$.  Since $f(x) < x$, arc $(g(x), x)$ is on the cycle.  By Lemma \ref{lemma:dfs}, all vertices on the cycle are descendants of $x$ in $F$.  Let $u$ be the first vertex after $x$ on the cycle such that $\mathit{sd}(u) < x$ and $x$ and $u$ are the same color.  Vertex $x$ is a candidate for $u$, so $u$ is well-defined.  Since $x$ and $u$ are the same color, $(g(u), u)$ is on the cycle.  If $g(u) < u$, then $\mathit{sd}(u) = g(u) \ge x$, contradicting the choice of $u$.  Thus $g(u) > u$.  Let $z$ be a vertex satisfying Lemma \ref{lemma:semidominator-equal-node} for $u$: $\mathit{sd}(z) = \mathit{sd}(u)$ and $z$ is on $F(a, g(u)]$, where $a$ is the nearest common ancestor of $u$ and $g(u)$ in $F$.  Since $z$ is a proper descendant of $a$, and $a$ is a descendant of $x$, $z$ is a proper descendant of $x$.  Since $\mathit{sd}(z) = \mathit{sd}(u) < x$, $x$ is a candidate for both $\mathit{rd}(z)$ and $\mathit{rd}(u)$.  Thus $\max\{\mathit{sd}(\mathit{rd}(z)), \mathit{sd}(\mathit{rd}(u))\} \le \mathit{sd}(x)$.  Every descendant $w$ of $x$ that is an ancestor of either $u$ or $z$ has a high path to $x$ through descendants of $z$, so $\mathit{sd}(w) \ge \mathit{sd}(x)$.  It follows that $\mathit{rd}(z) = \mathit{rd}(u)$ and that this vertex is an ancestor of $x$ in $F$.  Hence the coloring rule colors $z$ the same as $u$, and $x$.  That is, (iv) in Lemma \ref{lemma:3-vertices-color} holds.

On the part of the cycle from $x$ to $g(u)$, let $v$ be the last vertex whose predecessor on the cycle is not a descendant of $z$ in $F$; that is, all vertices on the cycle from $v$ to $g(u)$ (inclusive) are descendants of $z$, but not $v$.  There must be such a vertex since $g(u)$ is a descendant of $z$ in $F$ but $x$ is not.  Since $v$ follows $x$ on the cycle and precedes $u$, the choice of $u$ implies $\mathit{sd}(v) \ge x$.  Thus it cannot be the case that $v = z$, since then $\mathit{sd}(v) = \mathit{sd}(u) < x$.  Thus $v \not= z$.  Since $v$ is a descendant of $z$ in $F$, it is a proper descendant.  Since its predecessor on the cycle is not a descendant of $z$, its predecessor is not $f(v)$, but $g(v)$. Thus $v$ has the same color as $x$ and $z$.  The nearest common ancestor of $v$ and $g(v)$ in $F$ is a candidate for $\mathit{sd}(v)$. Since $z$ is an ancestor of $v$ but not $g(v)$, this nearest common ancestor is less than $z$, which implies $\mathit{sd}(v) < z$.  As observed above, the choice of $u$ implies $\mathit{sd}(v) \ge x$.  Thus $\mathit{sd}(v)$ is a descendant of $x$ and a proper ancestor of $z$ in $F$.

Let $y = \mathit{rd}(v)$.  Since $z$ is an ancestor of $v$, $z$ is a candidate for $y$, so $\mathit{sd}(y) \le \mathit{sd}(z) < x \le \mathit{sd}(v)$.  The coloring rule gives $y$ the color not given to $v$, $x$, and $z$, hence $y$ is not $x$ or $z$.  We showed above that $x \not= z$.  Vertex $y$ is a descendant of $x$ in $F$, as is $z$.  Furthermore $y$ and $z$ are related in $F$, since both are ancestors of $v$.  Hence (i) in Lemma \ref{lemma:3-vertices-color} holds.  Since $z$ is a candidate for $\mathit{rd}(v) = y$, it cannot be the case that $\mathit{sd}(y) = \mathit{sd}(z)$, or $y$ and $z$ would have the same color by Lemma \ref{lemma:2-vertices-color}.  Hence $\mathit{sd}(y) < \mathit{sd}(z)$.  For every vertex $w$ on $F(x, v]$, there is a high path from $w$ to $x$, consisting of $F(w, v]$ followed by the part of the cycle from $v$ to $x$.  It follows that $\mathit{sd}(w) \ge \mathit{sd}(x)$ for every such $w$.  Thus (iii) in Lemma \ref{lemma:3-vertices-color} holds.  Since $\mathit{sd}(w) \ge \mathit{sd}(x)$ for every vertex $w$ on $F[x, y]$ and $x$ and $y$ are different colors, $\mathit{sd}(x) \not= \mathit{sd}(y)$ by Lemma \ref{lemma:2-vertices-color}, so $\mathit{sd}(x) < \mathit{sd}(y)$.  We have previously established the other inequalities in (ii) of Lemma \ref{lemma:3-vertices-color}, so (ii) holds as well.  We conclude that $x$, $y$, $z$ is a violating triple.
\end{proof}

To prove that $B$ and $R$ are independent, we need two more definitions and one more technical lemma.  Let $v$ be a blue (red) vertex. Since $s$ is uncolored, $v \not= s$.  The \emph{pseudo-semi-dominator} $\mathit{psd}(v)$ of $v$ is the nearest ancestor of $v$ in $B$ $(R)$ that is less than $v$.  Such a vertex exists since $B$ and $R$ are spanning trees rooted at $s$.  By Lemma \ref{lemma:dfs}, $\mathit{psd}(v)$ is an ancestor of both $v$ and $g(v)$ in $F$.  (Arc $(g(v), v)$ is in the tree defining $\mathit{psd}(v)$.)  Furthermore $\mathit{psd}(v) \ge \mathit{sd}(v)$, since the path from $\mathit{psd}(v)$ to $v$ in the tree defining $\mathit{psd}(v)$ is a high path.  Thus $\mathit{psd}(v)$ is on $F[\mathit{sd}(v), a]$, where $a$ is the nearest common ancestor of $v$ and $g(v)$ in $F$.  The \emph{pseudo-relative-dominator} $\mathit{prd}(v)$ of $v$ is the vertex $x$ on $F(\mathit{psd}(v), v]$ such that $\mathit{sd}(x)$ is minimum, with a tie broken in favor of the smallest $x$. If $\mathit{prd}(v) \not= \mathit{rd}(v)$, then $\mathit{psd}(v) > \mathit{sd}(v)$ and $\mathit{rd}(v)$ is on $F(\mathit{sd}(v), \mathit{psd}(v)]$.  See Figure \ref{fig:psd-prd}.

\begin{figure}
\begin{center}
\scalebox{1.}[1.]{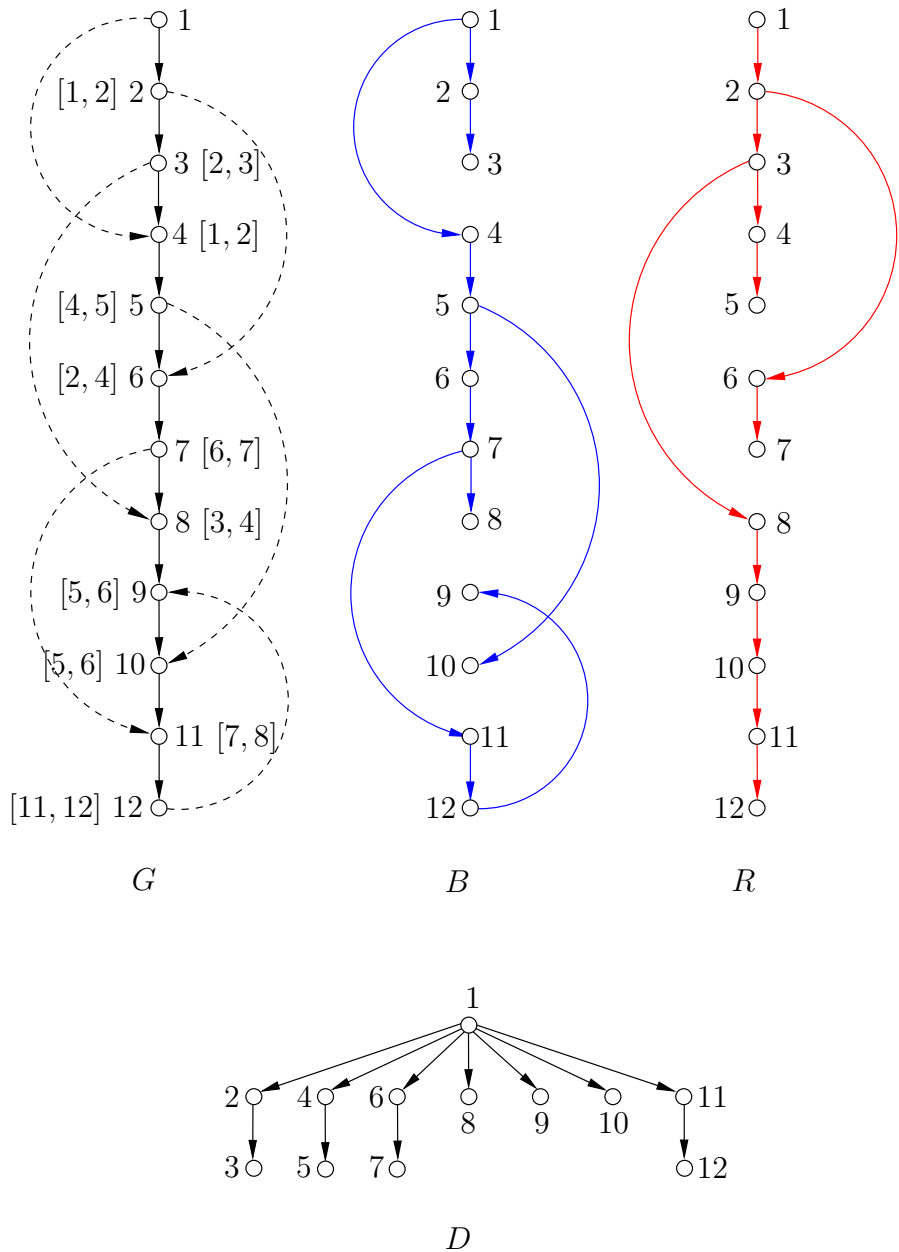}
\end{center}
\caption{A situation where $\mathit{sd}(v) \not= \mathit{psd}(v)$ and $\mathit{rd}(v) \not= \mathit{prd}(v)$, for the blue vertex $v=9$; we have $\mathit{sd}(9)=5$, $\mathit{rd}(9)=6$, $\mathit{psd}(9)=7$ and $\mathit{prd}(9)=8$. Vertices are numbered in preorder
with respect to a depth-first spanning tree shown with solid arcs; dashed arcs are not in the tree. Numbers inside the brackets correspond to
$\mathit{sd}(v)$ and $\mathit{rd}(v)$.\label{fig:psd-prd}}
\end{figure}

\begin{lemma}
\label{lemma:pseudo-semi-dominator}
For any vertex $v \not= s$, either $\mathit{psd}(v) = \mathit{sd}(v)$, or $\mathit{sd}(\mathit{prd}(v)) < \mathit{sd}(v)$ and $v$ and $\mathit{prd}(v)$ are different colors.
\end{lemma}
\begin{proof}
Suppose the lemma is false; let $v$ be the largest vertex that violates the lemma.  We prove by a case analysis that the lemma holds for $v$, a contradiction.  The last case produces an instance of Lemma \ref{lemma:3-vertices-color}.

By the choice of $v$, $\mathit{psd}(v) > \mathit{sd}(v)$. This implies $g(v) > v$: if $g(v) < v$, $\mathit{psd}(v) = g(v) = \mathit{sd}(v)$.  Let $a$ be the nearest common ancestor of $v$ and $g(v)$ in $F$.  Assume $\mathit{rd}(v)$ is a proper descendant of $a$ and $\mathit{sd}(\mathit{rd}(v)) < \mathit{sd}(v)$.  Since $\mathit{psd}(v)$ is an ancestor of $a$ by Lemma \ref{lemma:dfs}, $\mathit{rd}(v)$ a proper descendant of $a$ implies $\mathit{prd}(v) = \mathit{rd}(v)$.  Furthermore $\mathit{sd}(\mathit{prd}(v)) = \mathit{sd}(\mathit{rd}(v)) < \mathit{sd}(v)$, so the coloring rule colors $v$ and $\mathit{prd}(v) = \mathit{rd}(v)$ differently, which means that the lemma holds for $v$.  Hence we can assume that $\mathit{rd}(v)$ is an ancestor of $a$ or $\mathit{sd}(\mathit{rd}(v)) = \mathit{sd}(v)$.

Each vertex $u$ on $F(a, g(v)]$ has a high path to $v$, so $\mathit{sd}(u) \ge \mathit{sd}(v)$.  By Lemma \ref{lemma:semidominator-equal-node} there is such a vertex $u$ on $F(a, g(v)]$ with $\mathit{sd}(u) = \mathit{sd}(v)$.  Since $u$ is on $F(a, g(v)]$, $u > v$.  We prove that $u$ and $v$ have the same color.  If $\mathit{rd}(v)$ is an ancestor of $a$, then $\mathit{rd}(u) = \mathit{rd}(v)$, so the coloring rule colors $u$ and $v$ the same.  Suppose on the other hand that $\mathit{sd}(\mathit{rd}(v)) = \mathit{sd}(v)$.  Then $v$ is blue by the coloring rule.  Also, $\mathit{sd}(w) \ge \mathit{sd}(v)$ for every $w$ on $F(\mathit{sd}(v), a]$.  It follows that $\mathit{sd}(\mathit{rd}(u)) = \mathit{sd}(v) = \mathit{sd}(u)$, so $u$ is also blue by the coloring rule.  By the previous paragraph these are the only cases.

Suppose $u$ and $v$ are blue; the symmetric argument applies if they are red.  Let $w$ be the nearest ancestor of $g(v)$ in $B$ such that $w$ but not $b(w)$ is a descendant of $u$ in $F$.  There is such a $w$ since $g(v)$ is a descendant of $u$ in $F$ but $s$ is not.  Vertex $w$ is blue, since if $w \not= u$, $f(w)$ is a descendant of $u$ in $F$, so $b(w) = g(w)$ by the choice of $w$. Also $w \ge u > v$. Since $v$ is the maximum vertex that violates the lemma, the lemma holds for $w$.

There is a high path from $w$ to $v$, so $\mathit{sd}(w) \ge \mathit{sd}(v)$. Assume $\mathit{psd}(w) = \mathit{sd}(w) = \mathit{sd}(v)$.  Then $\mathit{psd}(w) < v$ and all vertices on $B(\mathit{psd}(w), w]$ are no less than $w$ and hence no less than $v$.  This is also true of the vertices on $B(w, g(v)]$, since they are all descendants of $u$, which is greater than $v$. It follows that $\mathit{psd}(v) = \mathit{psd}(w)$.  But then $\mathit{psd}(v) = \mathit{sd}(v)$, so the lemma holds for $v$.  Hence we can assume that $\mathit{psd}(w) > \mathit{sd}(w)$ or $\mathit{sd}(w) > \mathit{sd}(v)$.

We prove that $\mathit{prd}(w)$ is red, and that $\mathit{sd}(\mathit{prd}(w)) < \mathit{sd}(v)$. If $\mathit{sd}(w)=\mathit{sd}(v)$, both are true since $\mathit{psd}(w) > \mathit{sd}(w)$ by the previous paragraph, and the lemma holds for $w$, which is blue.
Thus assume $\mathit{sd}(w) > \mathit{sd}(v)$.
By Lemma \ref{lemma:dfs}, $\mathit{psd}(w)$ is an ancestor in $F$ of every vertex on $B[\mathit{psd}(w), w]$, including $w$ and $g(w)$.  Since $w$ but not $g(w)$ is a descendant of $u$ in $F$,  $u$ must be on $F(\mathit{psd}(w), w]$, making $u$ a candidate for $\mathit{prd}(w)$.  Thus $\mathit{sd}(\mathit{rd}(w)) \le \mathit{sd}(\mathit{prd}(w)) \le \mathit{sd}(u) = \mathit{sd}(v) < \mathit{sd}(w)$, so $\mathit{rd}(w)$ is red by the coloring rule, since $w$ is blue.
But then $\mathit{prd}(w)$ is also red, since $\mathit{psd}(w) = \mathit{sd}(w)$ implies $\mathit{prd}(w) = rd(w)$, and $\mathit{psd}(w) > \mathit{sd}(w)$ implies $\mathit{prd}(w)$ is red since the lemma holds for $w$. It cannot be true that $\mathit{sd}(\mathit{prd}(w)) = \mathit{sd}(v)$, for then $\mathit{sd}(\mathit{prd}(w)) = \mathit{sd}(v) = \mathit{sd}(u)$, making $\mathit{prd}(w)$ and $u$ the same color by Lemma \ref{lemma:2-vertices-color}, a contradiction.  Thus $\mathit{sd}(\mathit{prd}(w)) < \mathit{sd}(v)$.

If $\mathit{prd}(v) = \mathit{prd}(w)$, then $\mathit{prd}(v) = \mathit{prd}(w)$ is red and $\mathit{sd}(\mathit{prd}(v)) = \mathit{sd}(\mathit{prd}(w)) < \mathit{sd}(v)$ by the paragraph above, so the lemma holds for $v$.  Thus assume $\mathit{prd}(v) \not= \mathit{prd}(w)$.  Every vertex $x$ on $F(a, w]$ has a high path to $v$ via the path in $F$ to $w$ followed by the path in $B$ from $w$ to $v$, so $\mathit{sd}(x) \ge \mathit{sd}(v)$.  Since $\mathit{sd}(\mathit{prd}(w)) < \mathit{sd}(v)$, $\mathit{prd}(w)$ is not on $F(a, w]$, which means it is an ancestor of $a$ in $F$, and $\mathit{psd}(w)$ is a proper ancestor of $a$ in $F$.  This implies $\mathit{psd}(v) =\mathit{psd}(w)$, since $\mathit{psd}(w) < a \le v$ and every vertex on $B(\mathit{psd}(w), g(w)]$ is greater than $w$, which is greater than $v$.  Since $\mathit{prd}(v) \not= \mathit{prd}(w)$, $\mathit{prd}(v)$ must be on $F(a, v]$ and $\mathit{sd}(\mathit{prd}(v)) < \mathit{sd}(\mathit{prd}(w))$.  But $\mathit{sd}(\mathit{prd}(w)) < \mathit{sd}(v)$, so $\mathit{sd}(\mathit{prd}(v)) < \mathit{sd}(v)$.  It remains to prove that $\mathit{prd}(v)$ is red.

Since $\mathit{sd}(\mathit{rd}(v)) \le \mathit{sd}(\mathit{prd}(v)) < \mathit{sd}(v)$, $\mathit{rd}(v)$ is red by the coloring rule.  If $\mathit{sd}(\mathit{rd}(v)) = \mathit{sd}(\mathit{prd}(v))$, then $\mathit{prd}(v)$ is red by Lemma \ref{lemma:2-vertices-color}.  Thus assume $\mathit{sd}(\mathit{rd}(v)) < \mathit{sd}(\mathit{prd}(v))$.  We prove that $x = \mathit{rd}(v)$, $y = \mathit{prd}(v)$, $z = \mathit{prd}(w)$ satisfies the hypothesis of Lemma \ref{lemma:3-vertices-color}.  Indeed, $\mathit{sd}(\mathit{rd}(v)) < \mathit{sd}(\mathit{prd}(v)) < \mathit{sd}(\mathit{prd}(w)) < \mathit{sd}(v) < \mathit{rd}(v)$, so (ii) in the hypothesis of Lemma \ref{lemma:3-vertices-color} holds.  Since $\mathit{sd}(\mathit{rd}(v)) < \mathit{sd}(\mathit{prd}(v))$, $\mathit{rd}(v)$ is an ancestor of $\mathit{psd}(v) = \mathit{psd}(w)$ and hence an ancestor of both $\mathit{prd}(v)$ and $\mathit{prd}(w)$.  Also, $\mathit{prd}(v)$ and $\mathit{prd}(w)$ are related since $\mathit{prd}(w)$ is an ancestor of $a$.  Thus (i) holds.  Since $\mathit{prd}(v)$ and $\mathit{prd}(w)$ are ancestors of $v$ and descendants of $\mathit{rd}(v)$, (iii) holds by the definition of $\mathit{rd}(v)$.  We proved above that $\mathit{rd}(v)$ and $\mathit{prd}(w)$ are red, so (iv) holds.  By Lemma \ref{lemma:3-vertices-color}, $\mathit{prd}(v)$ is red.  Thus the lemma holds for $v$.
\end{proof}

\begin{theorem}
\label{theorem:BR-independent}
Trees $B$ and $R$ are independent.
\end{theorem}
\begin{proof}
The proof is like that of Lemma \ref{lemma:pseudo-semi-dominator}: we assume the lemma is false and work our way through a number of cases, each of which leads to a contradiction.  Two cases produce instances of Lemma \ref{lemma:3-vertices-color}.

Thus suppose the lemma is false.  Let $v$ be minimum such that $B[s, v]$ and $R[s, v]$ share a vertex other than a dominator of $v$.  By the argument in the proof of Theorem \ref{theorem:B-R-strong}, $B[d(v), v]$ and $R[d(v), v]$ share a vertex $w$ other than $v$ and $d(v)$.  Let $x_B$ and $x_R$ be the minimum vertices on $B[w, v]$ and $R[w, v]$, respectively. Assume $x_B \le x_R$; the symmetric argument applies if $x_R < x_B$.  By Lemma \ref{lemma:dfs}, $v$ is a descendant of both $x_B$ and $x_R$ in $F$, so $x_B$ is an ancestor of $x_R$ in $F$.

Let $u$ be a vertex of minimum $\mathit{sd}(u)$ on $F(x_B, v]$.  We investigate the properties of $u$.  First we prove that $\mathit{sd}(u)$ is a proper ancestor of $x_B$ in $F$.  Since $B[d(v), v]$ is simple, $x_B \not = d(v)$, which implies that $x_B$ does not dominate $v$; if it did, it would equal $d(v)$.  But $d(v)$ dominates $x_B$, since it dominates $v$.  Thus $d(v)$ is a proper ancestor of $x_B$ in $F$, and $d(v) < x_B$.  Assume $\mathit{sd}(u) \ge x_B$.   Lemma \ref{lemma:semi-relative-dominators} and the definition of $u$ imply that $d(u) = \mathit{sd}(u)$.  Furthermore an induction on $x$ for $x \in F[u, v]$ using Lemma \ref{lemma:semi-relative-dominators} shows that $d(x) \ge \mathit{sd}(u)$.  Thus $d(v) \ge \mathit{sd}(u) \ge x_B$, a contradiction.  We conclude that $\mathit{sd}(u) < x_B$, and $\mathit{sd}(u)$ is a proper ancestor of $x_B$ in $F$.

Second we prove that $\mathit{psd}(u)$ is also a proper ancestor of $x_B$ in $F$.  This is immediate if $\mathit{psd}(u) = \mathit{sd}(u)$.  If not, $\mathit{sd}(u) < \mathit{psd}(u)$, which implies by the choice of $u$ that $\mathit{prd}(u)$ is an ancestor of $x_B$ in $F$; hence $\mathit{psd}(u)$ is a proper ancestor of $x_B$ in $F$.

Third we prove that $u$ is red.  Suppose to the contrary that $u$ is blue.  If $u$ were on $B[x_B, v]$, $\mathit{psd}(u)$ would also be on $B[x_B, v]$, since $x_B < u$, but this contradicts $\mathit{psd}(u) < x_B$.  Thus $u$ is not on $B[x_B, v]$.  Let $x$ be the first vertex on $B[x_B, v]$ that is also on $F(u, v]$.  Since $f(x)$ is on $F[u, v]$, $b(x) = g(x)$, so $x$ is blue.  This implies $\mathit{psd}(x)$ is on $B[x_B, x]$, so $\mathit{psd}(x)$ is on $F[x_B, u]$.  The definition of $u$ gives $\mathit{prd}(x) = u$.  Since $x$ is blue, if $\mathit{psd}(x) = \mathit{sd}(x)$ then $x = \mathit{prd}(x) = \mathit{rd}(x)$ is red by the coloring rule; if $\mathit{psd}(x) > \mathit{sd}(x)$, then $u = \mathit{prd}(x)$ is red by Lemma \ref{lemma:pseudo-semi-dominator}.

Fourth we prove that $u$ is an ancestor of $x_R$ in $F$.  If not, $u$ is a proper descendant of $x_R$ in $F$.  But then an argument symmetric to the one in the previous paragraph shows that $u$ is blue, a contradiction.  Thus $u$ is an ancestor of $x_R$ in $F$.

Now we consider the relationship of $u$ and $x_B$, our goal being to construct an instance of Lemma \ref{lemma:3-vertices-color} in which $y = x_B$ and $z = u$ or vice-versa.  First we prove that $x_B$ is blue.  Since $u$ is on $F(x_B, x_R]$, $x_B \not= x_R$.  If $x_B = w$, $x_R = w = x_B$ since $x_B \le x_R \le w$, a contradiction.  Thus $x_B \not= w$.  The choice of $x_B$ implies $b(x_B) > x_B$, so $x_B$ is blue.

Second we prove $\mathit{sd}(u) \le \mathit{psd}(x_B)$.  Vertex $w$ is on $B[\mathit{psd}(x_B), x_B]$.  The path $P$ consisting of $B[\mathit{psd}(x_B), w]$ followed by $R[w, v]$ avoids $x_B$: $B[w, v]$ is simple, so $B[\mathit{psd}(x_B), w]$ avoids $x_B$, and $x_B < x_R$, so $x_B$ is not on $R[w, v]$.  Let $y$ be the first vertex on $P$ that is also on $F(x_B, v]$. The part of $P$ from $\mathit{psd}(x_B)$ to $y$ is a high path, so $\mathit{sd}(y) \le \mathit{psd}(x_B)$.  The choice of $u$ implies $\mathit{sd}(u) \le \mathit{sd}(y) \le \mathit{psd}(x_B)$.

Third we prove $\mathit{sd}(u) < \mathit{sd}(x_B)$.  If $\mathit{sd}(u) = \mathit{sd}(x_B)$, $u$ and $x_B$ are the same color by Lemma \ref{lemma:2-vertices-color}, a contradiction.  If $\mathit{sd}(u) > \mathit{sd}(x_B)$, then $\mathit{psd}(x_B) > \mathit{sd}(x_B)$ by the result in the previous paragraph.  By Lemma \ref{lemma:pseudo-semi-dominator}, $\mathit{sd}(\mathit{prd}(x_B)) < \mathit{sd}(x_B)$ and $\mathit{prd}(x_B)$ is red.  But then $x = \mathit{prd}(x_B)$, $y = x_B$, $z = u$ satisfies the hypothesis of Lemma \ref{lemma:3-vertices-color}, so $x_B$ is red, a contradiction.  The only remaining possibility is $\mathit{sd}(u) < \mathit{sd}(x_B)$.

We conclude the proof by showing that $x = \mathit{prd}(x_R)$, $y = u$, $z = x_B$ satisfies the hypothesis of Lemma \ref{lemma:3-vertices-color}, so $u$ is blue, a contradiction.  This requires some additional work.  Suppose $w$ is a descendant of $u$ in $F$.  Then the path from $\mathit{sd}(u)$ to $x_B$ consisting of the high path from $\mathit{sd}(u)$ to $u$ followed by $F[u,w]$ followed by $B[w, x_B]$ is a high path for $x_B$, giving $\mathit{sd}(x_B) \le \mathit{sd}(u)$, contradicting the result in the previous paragraph.  Thus $w$ is not a descendant of $u$ in $F$.

Since $x_R$ is a descendant of $u$ in $F$, $x_R \not= w$.  This implies $r(x_R) > x_R$, so $x_R$ is red.  Since $w$ is on $R[\mathit{psd}(x_R), x_R]$, $w$ is a descendant of $\mathit{psd}(x_R)$ by Lemma \ref{lemma:dfs}, so $u$ cannot be an ancestor of $\mathit{psd}(x_R)$ in $F$, or $w$ would be a descendant of $u$.  Thus $\mathit{psd}(x_R)$ is a proper ancestor of $u$, making $u$ a candidate for $\mathit{prd}(x_R)$.  This implies $\mathit{sd}(\mathit{prd}(x_R)) \le \mathit{sd}(u)$ and $\mathit{prd}(x_R)$ is an ancestor of $u$, so (i) and (iii) in the hypothesis of Lemma \ref{lemma:3-vertices-color} hold.  We proved the first inequality of (ii) above.  The path $R[\mathit{psd}(x_R), w]$ followed by $B[w, x_B]$ is a high path, so $\mathit{sd}(x_B) \le \mathit{psd}(x_R) < \mathit{prd}(x_R)$, giving the third inequality of (ii).

We prove that $\mathit{prd}(x_R)$ is blue, so (iv) holds.  If $\mathit{psd}(x_R) > \mathit{sd}(x_R)$, then $\mathit{prd}(x_R)$ is blue by Lemma \ref{lemma:pseudo-semi-dominator}.  Otherwise $\mathit{psd}(x_R) = \mathit{sd}(x_R)$.  This implies $\mathit{prd}(x_R) = \mathit{rd}(x_R)$, giving $\mathit{sd}(\mathit{rd}(x_R)) = \mathit{sd}(\mathit{prd}(x_R)) \le \mathit{sd}(u) < \mathit{sd}(x_B) \le \mathit{psd}(x_R) = \mathit{sd}(x_R)$, so $\mathit{prd}(x_R) = \mathit{rd}(x_R)$ is blue by the coloring rule.

Finally, since $\mathit{prd}(x_R)$ is an ancestor of $u$, $\mathit{sd}(\mathit{prd}(x_R)) = \mathit{sd}(u)$ would imply that $u$ is blue by Lemma \ref{lemma:2-vertices-color}, a contradiction.  Thus $\mathit{sd}(\mathit{prd}(x_R)) < \mathit{sd}(u)$, giving the second inequality of (ii).  By Lemma \ref{lemma:pseudo-semi-dominator}, $u$ is blue, the final contradiction.
\end{proof}

Summarizing the results of this section, we can construct a low-order for an arbitrary graph from its semi-dominator information by building a pair of independent spanning trees using Algorithm 8, and then constructing a low-high order using Algorithm 7.  This takes $O(n)$ time.  Although the algorithms are simple, the correctness proof of Algorithm 8 is quite complicated.

\section{Applications and Open Problems}
\label{sec:remarks}

We conclude by mentioning some further applications of our constructions and some open problems. We begin with applications.

\ignore{
In Section \ref{sec:two-independent-spanning-trees} we mentioned a conjecture that the two independent spanning trees constructed by Algorithm 8 may actually be strongly independent. It is also possible that they are constructable by Algorithm 1. This last conjecture (which subsumes strong independence) is supported by experiments with large graphs. We will present experimental results that show the practical efficiency of the algorithms we report here in a forthcoming paper.
}

The problem of testing the 2-vertex connectivity of a directed graph can be reduced in linear time
to verifying that the dominator tree of a flow graph is flat~\cite{2vc,Italiano2012}. Thus the remark in Section \ref{sec:low-high-headers} gives a simple linear-time algorithm for this problem. Our linear-time algorithms for computing two (strongly) independent spanning trees can be used in a data structure that computes pairs of vertex-disjoint $s$-$t$ paths in $2$-vertex connected digraphs (for any two query vertices $s$ and $t$) ~\cite{2vc}, and in fast algorithms for approximating the smallest 2-vertex connected spanning subgraph of a directed graph~\cite{2VCSS:Geo}.

Tholey~\cite{Tholey2012} considers the following problem: Let $(u_i,v_i)$, $1 \le i \le k$, be $k$ pairs of vertices of a directed acyclic graph with two distinguished start vertices $s_1$ and $s_2$. For each pair $(u_i,v_i)$, we wish to test if there are two vertex disjoint paths $P_1=(s_1,\ldots,t_1)$ and $P_2 = (s_2,\ldots,t_2)$, where $\{ t_1, t_2 \} = \{u_i, v_i \}$, and to construct such paths if they exist.
Tholey showed how to test the existence of $P_1$ and $P_2$ in constant time and how to produce them in $O(|P_1| + |P_2|)$ time after linear-time preprocessing.
He then uses this result to give a linear-time algorithm for the $2$-disjoint paths problem on a directed acyclic graph.

\ignore{
This result is based on algorithm of Suurballe and Tarjan~\cite{suurballe-tarjan} for finding shortest pairs of disjoint paths in a directed graph, and is used in \cite{Tholey05} to give an efficient algorithm for the $2$-disjoint paths problem on a directed acyclic graph.
}

Tholey's algorithm for testing the existence of $P_1$ and $P_2$ and for constructing them uses dominator trees, shortest-path trees, a topological order of the directed acyclic graph, and other structures.
The use of a low-high order gives us an alternative solution that works for a general directed graph $G$. We add to $G$ a new root vertex $s$ together with arcs $(s,s_1)$ and $(s,s_2)$, and compute the following structures: the dominator tree $D$, a low-high order of $D$, and two strongly independent spanning trees $B$ and $R$ as described in Section \ref{sec:properties}. Also, for each vertex $v \not= s$, we store the child $c(v)$ of $s$ in $D$ that dominates $v$. All computations take linear time. Let $u,v \not= s$. We can test if $G$ contains the above vertex-disjoint paths $P_1$ and $P_2$ as follows. If $u=v$, then the paths exist if and only if $c(u)=u$. Otherwise, the paths exist if and only if $c(u) \not= c(v)$. This test takes constant time. If the paths exist, we can produce them in $O(|P_1| + |P_2|)$ time using $B$ and $R$ as mentioned in the remark in Section \ref{sec:properties}.

Now we turn to open problems.
A conjecture, attributed to Frank in \cite{vdp:whitty},
states that any strongly $k$-connected graph contains $k$
disjoint branchings. This vertex-disjointness
conjecture is analogous to a well-known theorem of
Edmonds on edge-disjoint branchings~\cite{matroids:edmonds}.
Edmonds's result states that $G$ has $k$ edge-disjoint branchings
rooted at $s$ if and only if, for every vertex $v$, there are $k$
edge-disjoint paths from $s$ to $v$. (See also
\cite{edge-disjoint:edmonds} for a different characterization of
edge-disjoint branchings.) The vertex-disjointness conjecture was disproved in 1995
by Huck~\cite{ind_st:huck95}, who showed that for any $k \ge 3$
there is a $k$-connected graph that does not have $k$ independent
branchings. This disproof of the conjecture does not hold in
special cases. For instance, it does not apply to planar
graphs~\cite{ind_st:huck99}.
In acyclic graphs a statement related to the vertex-disjointness conjecture is true.
Specifically, let $G$ be an acyclic directed graph that is $k$-connected to
a vertex $t$: for each $v \neq t$, there are $k$
vertex-disjoint paths from $v$ to $t$. Then $G$ has $k$ independent
in-trees rooted at $t$~\cite{ind_acyclic_st:huck99}. The
same result was proved independently
in~\cite{acyclic_independence:abhs00}; there, the authors gave an
algorithm that constructs the $k$ trees in $O(k^2n + km)$ time,
starting from a topological order of $G$.
Itai and Rodeh made an analogous vertex-disjointness conjecture
for undirected graphs \cite{ind_st:ir84, ind_st:ir}. Specifically
they conjectured that for any $k$-connected undirected graph
$G=(V,E)$ and for any vertex $v \in V$, $G$ has $k$ independent
spanning trees rooted at $v$. Itai and Rodeh proved their conjecture
for the case $k=2$, and gave a linear-time construction. The case
$k=3$ was proved by Cheriyan and Maheshwari~\cite{ind_st:cm}, who
also gave a corresponding $O(n^2)$-time algorithm, and by Itai and
Zehavi~\cite{ind_st:iz}. Curran, Lee and
Yu~\cite{ind_st:cly06} provided a $O(n^3)$-time algorithm that
constructs four independent spanning trees of a 4-connected graph,
thus proving the $k=4$ case. To the best of our knowledge, the case $k \ge 5$ is open.

Algorithm 8 raises a couple of open problems:
\begin{itemize}
\item[(1)] Are the trees $B$ and $R$ constructed by Algorithm 8 strongly independent?  The example in Figure \ref{fig:independent-strong} shows that independent trees are not in general strongly independent, but the specific trees constructed by Algorithm 8 might be. If so, Algorithm 8 gives a direct construction of such trees, avoiding the need to use Algorithms 7 and 1.
\item[(2)] Can the trees $B$ and $R$ constructed by Algorithm 8 be generated from some low-high order of $D$ by Algorithm 1?  If this is true then (1) is true as well.  Furthermore, the following would be an alternative to Algorithm 7: Apply Algorithm 8 to construct spanning trees $B$ and $R$. Construct the corresponding spanning trees $B'$ and $R'$ in the derived graph. Form the graph $K$ whose arcs are all those in $B'$ but not in $D$ and the reversals of all arcs in $R'$ that are not in $D$. Find a topological order of $K$. For each vertex $v \not= s$, arrange its children in $D$ in an order consistent with the topological order of $K$. The preorder on $D$ corresponding to its ordered lists of children will be a low-high order.
\end{itemize}
Experiments on large graphs suggest that (2) and hence (1) is true, but we have no proof.

Finally, is there a simple way to extend the iterative algorithm of Cooper et al~\cite{dom:chk01} or incremental algorithms for computing dominators~\cite{dyndom:2012,RR:incdom,SGL} so that they also compute a low-high order of the dominator tree?
Such an extension would make these into certifying algorithms.

\bibliographystyle{plain}
\bibliography{domcert}

\end{document}